\documentclass{amsart}
\usepackage{amsmath,amssymb,latexsym,graphicx,geometry,amsfonts,color,tikz,alltt}
\usetikzlibrary{shapes,snakes,calendar,matrix,backgrounds,folding}

\title{The Permanent, Graph Gadgets and counting solutions for certain types of planar formulas}

\author{Christian Schridde}
\address{Federal Office for Information Security\\D-53179 Bonn, Germany}
\email{ChristianSchridde@googlemail.de}

\newcommand{\sat}{\mathsf{SAT}}
\newcommand{\dreisat}{\mathsf{3SAT}}
\newcommand{\nestedsat}{\mathsf{NESTED}\text{-}\mathsf{SAT}}
\newcommand{\acyclicsat}{\mathsf{aCYCLIC}\text{-}\mathsf{SAT}}
\newcommand{\forestdreisat}{\mathsf{FOREST}\text{-}\mathsf{3SAT}}
\newcommand{\pdreisat}{\mathsf{PLANAR}\text{-}\mathsf{3SAT}}
\newcommand{\pnpksat}{\mathsf{PN}\text{-}\mathsf{PLANAR}\text{-}\mathsf{kSAT}}

\newcommand{\pnpdreisat}{\mathsf{PN}\text{-}\mathsf{PLANAR}\text{-}\mathsf{3SAT}}

\newcommand{\classBPP}{\bold{BPP}}
\newcommand{\usat}{\mathsf{Unique}\mathrm{-}\mathsf{SAT}}
\newcommand{\uksat}{\mathsf{Unique}\mathrm{-}\mathsf{kSAT}}
\newcommand{\udreisat}{\mathsf{Unique}\mathrm{-}\mathsf{3SAT}}

\newcommand{\W}{\operatorname{W}}
\newcommand{\classP}{\mathsf{P}}
\newcommand{\classNP}{\mathsf{NP}}
\newcommand{\classRP}{\mathsf{RP}}
\newcommand{\detm}{\text{det}}
\newcommand{\perm}{\text{perm}}
\newcommand{\define}{\stackrel{\text{def}}{=}}
\newcommand{\mA}{\bold{A}}
\newcommand{\mM}{\bold{M}}
\newcommand{\mB}{\bold{B}}

\newcommand{\grG}{\mathcal{G}}
\newcommand{\dji}{\mathsf{DJi}}
\newcommand{\tp}{\mathsf{t}}
\newcommand{\gad}{\mathfrak{G}}

\makeatletter
\renewcommand*\env@matrix[1][c]{\hskip -\arraycolsep
  \let\@ifnextchar\new@ifnextchar
  \array{*\c@MaxMatrixCols #1}}
\makeatother

\newtheorem{theorem}{Theorem}[section]
\newtheorem{lemma}[theorem]{Lemma}
\newtheorem{corollary}[theorem]{Corollary}
\newtheorem{conjecture}[theorem]{Conjecture}
\theoremstyle{definition}
\newtheorem{defin}[theorem]{Definition}
\theoremstyle{proposition}
\newtheorem{proposition}[theorem]{Proposition}
\begin{document}

\maketitle

\begin{abstract}
	In this paper, we build on the idea of Valiant \cite{Val79a} and Ben-Dor/Halevi \cite{Ben93}, that is, to count the number of satisfying solutions of a boolean formula via computing the permanent of a specially constructed matrix. We show that the Desnanot-Jacobi identity ($\dji$) prevents Valiant's original approach to achieve a parsimonious reduction to the permanent over a field of characteristic two. As the next step, since the computation of the permanent is $\#\classP$-complete, we make use of the equality of the permanent and the number of perfect matchings in an unweighted graph's bipartite double cover. Whenever this bipartite double cover (BDC) is planar, the number of perfect matchings can be counted in polynomial time using Kasteleyn's algorithm \cite{Kas67}. To enforce planarity of the BDC, we replace Valiant's original gadgets with new gadgets and describe what properties these gadgets must have. We show that the property of \textit{circular planarity} plays a crucial role to find the correct gadgets for a counting problem. To circumvent the $\dji$-barrier, we switch over to fields $\mathbb{Z}/p\mathbb{Z}$, for a prime $p > 2$. 

With this approach we are able to count the number of solutions for $\forestdreisat$ formulas in randomized polynomial time.  Finally, we present a conjecture that states which kind of generalized gadgets can not be found, since otherwise one could prove $\classRP = \classNP$.  The conjecture establishes a relationship between the determinants of the minors of a graph $\grG$'s adjacency matrix and the \textit{circular planar} structure of $\grG$'s BDC regarding a given set of nodes.
\end{abstract}

\section{Introduction}

Given a $(n\times n)$-matrix $\mA = a_{ij}$, $1\leq i,j \leq n$, its determinant is defined as 
\begin{equation}
	\detm(\mA) = \sum_{\sigma \in S_n} \text{sgn}(\sigma)\prod^n_{i=1}a_{i,\sigma_i}
\end{equation}
and the permanent is defined as
\begin{equation}
	\perm(\mA) = \sum_{\sigma \in S_n}\prod^n_{i=1}a_{i,\sigma_i}
\end{equation}
The difference in notation is minimal. The determinant considers the sign of the permutations $\sigma$ whereof the permanent does not. 
This little difference is responsible for the huge gap between the necessary resources that are needed in order to compute the actual value of these functions. The determinant can be evaluated in time that is polynomial in the dimension $n$, whereof the permanent is only computable in time that is exponential in $n$, e.g., using the $\mathcal{O}(n2^n)$ algorithm of Ryser. In the year 1979 Leslie Valiant \cite{Val79a,Val79b} showed that the permanent is even $\#\classP$-complete and he shows the same for any computation of $\perm(\mA)\pmod{p}$, $p > 2$. In his proof, Valiant created a graph $\mathcal{G}$ from a given boolean formula $\Phi$ such that the number of $\Phi$'s satisfying assignments can be obtained from the permanent of $\mathcal{G}$'s adjacency matrix by using Proposition \ref{prop1}.

\begin{defin}[Cyclic-cover]
 A \textit{cyclic-cover} of a weighted directed graph $\grG = (V,E)$ is a subset $R \subseteq E$ that forms a collection of node-disjoint directed cycles that cover all the nodes of $\grG$.$\blacktriangleleft$
\end{defin}

\begin{defin}[Weight]
 The weight of a cyclic cover $R$, denoted by $\W(R)$, is the product of the weights of the edges in $R$.$\blacktriangleleft$
\end{defin}

\begin{proposition}\label{prop1}
 Given a graph $\grG$ and its adjacent matrix $\mA$. Let $\mathcal{R}$ be the set of all cyclic covers of $\grG$, then $\perm(\mA) = \sum_{R \in \mathcal{R}} \W(R)$.
\end{proposition}

The essential ingredients for Valiant's proof are \textit{graph gadgets}. These gadgets ensure that the weights of the cyclic covers depend on whether they are part of a satisfying assignment or not. In a subsequent work, Ben-Dor and Halevi \cite{Ben93} presented another, simpler proof that comes along with another type of graph gadget. 

\textbf{Contribution.} In this paper, we extent the idea of Valiant and Ben-Dor/Halevi by means of defining new gadgets. In Valiant's approach, his gadgets are designed to make the permanent of the graph's adjacency matrix contain the number of satisfying solutions of a boolean formula $\Phi$. Since the permanent is equal to the determinant in fields of characteristic two, we first show, that gadgets that would allow to count the number of satisfying solutions (at least of formulas from $\usat$) can not be found because of the Desnanot-Jacobi identity.

Since despite the case modulo $2$, the permanent can not be evaluated in general, we then turn our attention to the relationship of the permanent and the number of perfect matchings in a graph's bipartite double cover. We search for gadgets that have the potential to determine satisfiability over $\mathbb{Z}/p\mathbb{Z}$, for $p > 2$, whenever the bipartite double cover is planar. We constructed a new XOR-gadget\footnote{That is the name of one of the gadgets Valiant uses in his proof.} that has two mandatory properties:
\begin{itemize}
 \item It is an unweighted graph, i.e, its adjacency matrix is a $(0/1)$-matrix. This is in contrast to the XOR-gadget of Valiant (and the CLAUSE-gadget of Ben-Dor and Halevi) which contains values from $\{-1,0,1,2,3\}$. Having a $(0/1)$-graph is mandatory to apply Proposition \ref{propPermPerf}.
 \item It fulfills a weakened set of permanent rules\footnote{Valiant introduces a set of permanent rules that an XOR-gadget must fulfil in order to make his proof work.} over a small finite group  rather than the integers.
\end{itemize}
In [\cite{Val79a}, Lemma 3.3] Valiant showed how to transform a matrix $\mA$ with positive integers into another matrix $\mA'$ that only has values from $\{0,1\}$ and $\perm(\mA) = \perm(\mA')$. Since his XOR-gadget does contain edges with weight $-1$, he could not apply the transformation. However, Ben-Dor and Halevi \cite{Ben93} showed the following reduction chain:
\begin{equation*}
 \text{IntPerm} \propto \text{NoNegPerm} \propto \text{2PowersPerm} \propto \text{0/1-Perm}
\end{equation*}
whereof the permanent value is preserved along the chain. The drawback of their approach is, that this reduction chain increases the number of nodes heavily, which reduces the chances to reach planarity, which is mandatory to apply the Kasteleyn's counting algorithm. They did not consider to look at the planarity of their resulting graph nor they tried to further work with it.

Although so far our new XOR-gadget already gives the theoretical potential to determine satisfiability, it turns out that most of the time Kasteleyn's algorithm can not be applied\footnote{At least, we draw this conclusion from our implementation, that shows that for large formulas the bipartite double cover is almost always non-planar.}. The reason is, that the BDC is almost always non-planar. Consequently, we tried to find methods that increases the chance that the BDC will be planar.
Even though the bipartite double cover operation changes planarity in both ways, it turns out that planar graphs with low connectivity do more often have a planar bipartite covers. Following this fact, we introduce a method to make $\grG$ planar by the help of two further new XOR-gadgets and Lichtenstein's reduction \cite{Lic82} to pn-planar\footnote{A certain class of boolean formulas that is also NP-complete.} formulas. 

Finally, one notices that the fraction of planar BDCs is still negligible. To circumvent the problem of getting a planar bipartite double cover only accidentally, in the last sections of the paper we turn towards generalization. We describe what gadgets have to be used therewith the bipartite double cover will be \textit{always} planar if the input is a pn-planar formula. We then show which subclasses of these formulas can be counted with that approach, since appropriate gadgets exists. And, even more important, what prevents the generalized approach from counting solutions of arbitrary pn-planar formulas.

At the end, we formulate a conjecture which states what kind of gadgets could not be found since otherwise one could prove $\classRP = \classNP$. The conjecture establishes a relationship between the determinants of the minors of a graph $\grG$'s adjacency matrix and the \textit{circular planar} structure of $\grG$'s bipartite double cover regarding a certain set of nodes.

\subsection{Preliminaries and Definitions}

\begin{defin}[Decision problem: $\sat$] Given a boolean formula, determine its satisfiability. $\blacktriangleleft$ \end{defin}

\begin{defin}[Decision problem: $\usat$] Given a boolean formula that has \textit{at most} one satisfying assignment, determine its satisfiability. $\blacktriangleleft$ \end{defin}

\begin{defin}[Decision problem: k$\sat$] Given a boolean formula that is in CNF with at most $k$ literals per clause (kCNF), determine its satisfiability.  $\blacktriangleleft$ \end{defin}

\begin{defin}[Decision problem: $\uksat$] Given a boolean formula in kCNF that has \textit{at most} one satisfying assignment, determine its satisfiability. $\blacktriangleleft$ \end{defin}

\begin{defin}[Discrete counting problem: $\#_p$ k$\sat$] Given a boolean formula from k$\sat$, compute the number of satisfying solutions modulo $p$.  $\blacktriangleleft$ \end{defin}

For $p=2$, the $\#_p$ k$\sat$ problem is equivalent to its parity problem $\otimes$k$\sat$.

\begin{theorem}[Valiant-Vazarani \cite{Val85}]\label{theorem1}
 If there is a polynomial-time algorithm solving the decision problem $\usat$ than $\classNP = \classRP$.
\end{theorem}

The class $\usat$ handles boolean formulas that have at most a single satisfying assignment. Surprisingly, this does not help \textit{in general} to decide their satisfiability. In the proof of Theorem \ref{theorem1}, Valiant and Vazarani showed, that one can transform an arbitrary boolean formula $\Phi$ with $n$ variables, via an efficient computable function $f$, into another formula such that
\begin{eqnarray*}
 	\Phi \; \text{is satisfiable} & \Rightarrow  & \text{Pr}[ f(\Phi) \; \text{has a unique assignment}] \geq \frac{1}{8n} \\ 
	\Phi \; \text{is not satisfiable} & \Rightarrow & \text{Pr}[ f(\Phi) \; \text{is satisfiable}] = 0 
\end{eqnarray*}

Formulas from $\usat$ can represent problems from number theory or cryptology, e.g., it is possible to express the \textit{factorization problem} as the task to find the solution of a formula that only has one assignment and that involves the binary representation of the two factors.

\begin{defin}[Incidence graph of a formula]
 Let $\Phi$ be a CNF formula with variables $v_1,\ldots,v_n$ and clauses $c_1,\ldots,c_m$. Then the incidence graph $\mathcal{G}_\Phi$ is defined as
 \begin{align*}
	\mathcal{G}_\Phi & = (V,E) \\
	V 		 & = \{v_1,\ldots,v_n,c_1,\ldots,c_m\} \\
	E		 & = \{(v_i,c_j)| v_i\;\text{or}\;\neg v_i\;\text{occurs in}\;c_j\}
 \end{align*}
\end{defin}

\begin{defin}[$\pdreisat$]\label{defPLANAR3SAT}
 Given a 3CNF with a planar incidence graph (called ``a planar 3CNF formula''), then $\pdreisat$ is the decision problem to decide its satisfiability. $\blacktriangleleft$
\end{defin}

\begin{defin}[$\mathsf{POSITIVE}$-$\mathsf{NEGATIVE}$-$\pdreisat$]\label{defPNPLANAR3SAT}
 Let $\Phi$ be a 3CNF formula with variables $v_1,\ldots,v_n$ and clauses $c_1,\ldots,c_m$. The graph $\mathcal{G}_\Phi$ is defined as
  \begin{align*}
	\mathcal{G}_\Phi & = (V,E) \\
	V 		 & = \{v_1,\ldots,v_n,c_1,\ldots,c_m\} \\
	E		 & = E_1 \cup E_2\\
	E_1		 & = \{(v_i,v_{i+1})| 1 \leq i \leq n-1\} \cup (v_n,v_1)\\
	E_2		 & = \{(v_i,c_j)| v_i\;\text{or}\;\neg v_i\;\text{occur in}\;c_j\}
 \end{align*}
 If $\mathcal{G}_\Phi$ is planar, then $\Phi$ is a planar 3CNF formula. The edge set $E_1$ forms a cycle, i.e., it divides the plane into two faces,
 the inner and the outer face. The restriction in $\mathsf{POSITIVE}$-$\mathsf{NEGATIVE}$-$\pdreisat$ or short $\pnpdreisat$ is that, if $c_i$ and $c_j$ contain the same variable 
 but in a negated form, then $c_i$ and $c_j$ are in different faces of the graph. $\pnpdreisat$ is the decision problem to decide if $\Phi$ satisfiable. $\blacktriangleleft$
\end{defin}

The graph from Definition \ref{defPLANAR3SAT} is bipartite since it is a mapping between the two vertex sets $\{v_1,...,v_n\}$ and $\{c_1,...,c_m\}$, whereof the graph from Definition \ref{defPLANAR3SAT} is not bipartite directly, but can be converted into a bipartite graph by removing the edges from $E_1$. In his paper, Lichtenstein describes a reduction algorithm, that transforms an arbitrary formula $\Phi$ to another formula $\Phi'$, such that the graph $\mathcal{G}_{\Phi'}$ is (pn)-planar. This transformation may change the number of satisfying solutions of $\Phi$, but not $\Phi$'s satisfiability. 

\begin{defin}
 $\pnpksat$ uses Definition \ref{defPNPLANAR3SAT} but with a kCNF as the input.$\blacktriangleleft$
\end{defin}


\begin{defin}[$\forestdreisat$]\label{defForSat}
 Given a 3CNF formula, who incidence graph is a forest. The $\forestdreisat$ is the decision problem to decide if $\Phi$ satisfiable.$\blacktriangleleft$
\end{defin}

Surprisingly, there are special classes of satisfiability/counting problems that could be solved in polynomial time. In 2010, Marko Samer and Stefan Szeider \cite{Sam10b} already showed that $\#\forestdreisat$ can be computed in polynomial time using a approach that has nothing to do with graph gadgets. Dan Roth \cite{Rot96} describes $\acyclicsat$, which is a special case of monotone 2CNF formulas and showed, that its counting version  ($\#\mathsf{aCYCLIC}$-$\mathsf{2MON}$-$\mathsf{SAT}$) is polynomial time solvable. $\nestedsat$ is handled by Knuth \cite{Knu90} and Kratochvíl and Krivánek \cite{Kra93} and it is shown that $\nestedsat$ is actual in $\classP$. So far it is not known if the counting version is also efficiently solvable. 

Further class of counting algorithms build on planar graphs. The reason is, that in planar graphs some problems can be solved more efficiently than in non-planar graphs. One such problem is to count the number of perfect matchings. Fisher, Kasteleyn and Temperley \cite{Kas67} discovered an algorithm that finds the Pfaffian orientation of a planar graph in polynomial time, which can be used to efficiently compute the number of perfect matchings. As already mentioned in the beginning, it is well known that the permanent of a $(0/1)$-matrix $\mA$ is equal to the number of perfect matchings in the graph that belongs to $\mA$'s biadjacency matrix. Based on this relationship, Valiant invented the hole theory of \textit{Holographic Algorithms} \cite{Val06,Val08} which is extended, e.g., by works of Cai and Lu \cite{Cai07,Cai10}. These holographic algorithms lead to polynomial time computable algorithms whereof former only exponential ones where known, e.g., $\mathsf{\#X}$-$\mathsf{MATCHINGS}$,$\mathsf{\#PL}$-$\mathsf{3}$-$\mathsf{NAE}$-$\mathsf{SAT}$, $\mathsf{\#PL}$-$\mathsf{EVEN}$-$\mathsf{LIN}$-$\mathsf{2}$. In that course a surprising result is that $\#_2\mathsf{Pl}$-$\mathsf{RTW}$-$\mathsf{MON}$-$\mathsf{3CNF}$\footnote{This stands for \textit{Planar-ReadTwice-Monotone-3CNF}, whereof read twice means that each variable occurs at most two times and monotone means that each variable only occurs positive.} has been proved to be $\otimes \classP$-complete, but $\#_7\mathsf{Pl}$-$\mathsf{RTW}$-$\mathsf{MON}$-$\mathsf{3CNF}$ to be polynomial time computable. Later it is shown by Cai et al \cite{Cai07} that this is actually a special case of $\#_{2^k-1}\mathsf{Pl}$-$\mathsf{RTW}$-$\mathsf{MON}$-$\mathsf{3CNF}$, which are all polynomial time computable. The idea of holographic algorithms can be seen as a further development of Valiant's proof idea to reduce counting satisfying assignments to the permanent. He replaced the graph gadgets by a more general type called \textit{matchgate} but also showed in his paper \cite{Val06} that there do not exits elementary matchgrid algorithms for $3$CNF formulas.

\textbf{Notations:} With each graph $\grG$, we associated two matrices; first its adjacency matrix $\mA_\grG$ and the bipartite adjacency matrix $\mB_\grG \define \begin{pmatrix} 0 & \mA_\grG\\ \mA^\tp_\grG & 0 \end{pmatrix}$, which is always undirected and loop-free. $\mA^\tp$ is the transposition of the matrix $\mA$. Note that $\mB_\grG^\tp = \mB_\grG$. The graph of $\mB_\grG$ is denoted by $\hat{\grG}$ and is the \textit{bipartite double cover} graph of $\grG$.  Building the bipartite double cover of $\grG$ can also be written as the tensor graph product $\hat{\grG} = \grG \times K_2$ or as the Kronecker product of the two matrices $\mB_\grG = \mA_\grG \otimes \mA_{K_2}$.

With $\bold{0}^{n \times m}$ we denote a zero matrix with dimensions $n$ and $m$. The notation $\mA^{i_1,...,i_k}_{o_1,...,o_k}$ means the sub-matrix of $\mA$ that is created by removing the columns $i_1,...,i_k$ and the rows $o_1,...,o_k$. For a boolean formula $\Phi$, the number of satisfying assignments is denoted as $\#\Phi$. The term $\Phi_{x_i=1}$, means the formula $\Phi$ with $x_i$ substituted by $1$. The graph that is constructed from a boolean formula according to the idea described by Valiant, is called a $\textit{Valiant graph}$. 

As usual, we use $\mathsf{1}_{<\text{statement}>}$ to denote $1$ if the statement is true, and $0$ otherwise. E.g., $\mathsf{1}_{\Phi\text{ is sat}}$ is equal to $1$ if the formula $\Phi$ is satisfiable and $0$ otherwise. We write, e.g. $X = 3\cdot\mathbb{Z}$, if $X$ is a multiple of $3$.

\begin{defin}[Parsimonious reduction]
 $f$ is reducible to $g$ if there exists a polynomial time computable function $\rho: \{0,1\}^* \rightarrow \{0,1\}^*$ such that for every $x \in \{0,1\}^*$, $f(x) = g(\rho(x))$.$\blacktriangleleft$
\end{defin}

\begin{defin}[Many-one reduction]
  $f$ is reducible to $g$ if there exists a polynomial time computable function $\rho: \{0,1\}^* \rightarrow \{0,1\}^*$ and $\tau: \mathbb{N} \rightarrow \mathbb{N}$ such that for every $x \in \{0,1\}^*$, $f(x) = \tau(g(\rho(x)))$.$\blacktriangleleft$
\end{defin}

\begin{proposition}[Relationship: Permanent and PerfectMatchings]\label{propPermPerf}
 Let $\grG$ be an arbitrary unweighted graph and $\mA_\grG$ its $(0/1)$-adjacency matrix, then the permanent of $\mA_\grG$ is equal to the number of perfect matchings in the graph $\hat{\grG}$ (i.e., the bipartite double-cover).
\end{proposition}

Proposition \ref{propPermPerf} states a well known relationship between computing the permanent and counting the number of perfect matchings. The later is in general also a $\classNP$-hard problem, but can be solved efficiently if the graph is planar. 

\begin{theorem}[Fisher-Kasteleyn-Temperley (FKT) algorithm]\label{theoremKastelyen}
 If $\mA$ is a $(0/1)$-adjacency matrix of a \textit{planar graph} $\grG$, then the number of perfect matchings in $\grG$ can be computed in polynomial time.
\end{theorem}

The FKT algorithm uses the planarity of the graph to obtain an embedding in the plane. From this embedding, the algorithm orients the edges of all faces in a certain way. The orientation of all edges is called a \textit{Pfaffian orientation}. This orientations permits to count the number of perfect matchings efficiently.

\begin{proposition}\label{propCons}
 The permanent of a $(0/1)$-matrix $\mA$ with associated graph $\grG$, can be computed in polynomial time if $\hat{\grG}$ is planar.
\end{proposition}
Proposition \ref{propCons} is just a consequence from Proposition \ref{propPermPerf} and Theorem \ref{theoremKastelyen}. Finally, the planarity of a graph is defined as:

\begin{defin}[Planar graphs]
 A graph $\mathcal{G} = (V,E)$ is called \textit{planar} if it can be embedding in the plane without two edges cross each other. $\blacktriangleleft$
\end{defin}


\begin{defin}[Circular planar graphs]
 A graph $\mathcal{G} = (V,E)$ is called \textit{circular planar} according to a boundary $B \subseteq V$, if it is planar and can be embedding on a disc $D$ in the plane with all nodes from $B$
lie on the circle that bounds $D$ and all other nodes are in the interior of $D$. The \textit{order} of a circular planar graph is the order of the boundary nodes in clockwise direction. $\blacktriangleleft$
\end{defin}

\textbf{Remark: The permanent of minors.} As also used by Valiant and Ben-Dor/Halevi: if the permanent of a $(0/1)$-matrix counts the number of cyclic covers, then the permanent of $A^i_o$ is equal to the number of those ways from node $i$ to node $o$, where all nodes not part of the way can be covered by a cyclic cover.

\section{Proofs that the permanent is $\#\classP$-complete} 

Since our work heavily depends on the idea of Valiant, we shortly describe the idea of the proof and the used gadgets. Beside the two named proofs of Valiant and Ben-Dor/Halevi, recently, Scott Aaronson \cite{Aar11} provided another proof based on a linear-optic setup and quantum reasoning that shows also that the permanent is $\#\classP$-complete. 

Whenever a graph is drawn in the rest of the paper, whereof an edge has no shown weight, the weight is equal to $1$. Whenever an adjacency matrix is given to a graph, no edge is labeled since the weights could be seen from the matrix. 

\subsection{Valiant's proof.} Valiant's idea was, given a boolean 3CNF formula $\Phi$ with $n$ variables and $m$ clauses, to construct a graph $\grG$, such that the permanent of $\mA_\grG$ involves $\#\Phi$. In particular, in his proof the permanent is equal to $4^{3m}\#\Phi$. In order to construct the graph, he uses small subgraphs, called gadgets, that have special combinatorial properties. These properties cause that each valid assignment of $\Phi$ is equivalent to a set of cyclic covers that sum up to a total weight of $4^{3m}$. And each assignment that does not satisfy $\Phi$ is equal to a set of cyclic covers that sum up to a total weight of zero. Since the permanent is the sum over all weights of all cyclic covers, it is $\perm(\mA) = 4^{3m}\#\Phi$. For the rest, we refer to the rewritten proof of Valiant that can be found in Barak and Aroras book \cite{San09}.

There are in total three gadgets that Valiant uses to construct his graph $G$.

\subsubsection{The CLAUSE-gadget.} For each of the $m$ clauses in $\Phi$, there is a small subgraph named CLAUSE-gadget. 

\begin{figure}[h!]
\centering
\begin{minipage}[h]{50mm}
	\centering
	\begin{tikzpicture}
		[scale=.8]
		\node[circle, fill, draw, inner sep = 2pt] (n1) at (0,1){}; \node[] at (-0.2,0.8) {\tiny{$1$}};
			\node[rectangle, inner sep = 2pt] (xor12) at (1,0) {...};
		\node[circle, fill, draw, inner sep = 2pt] (n2) at (2,1){}; \node[] at (2.2,0.8) {\tiny{$2$}};
			\node[rectangle, rotate=-68, inner sep = 2pt] (xor23) at (2.5,2.4) {...};
		\node[circle, fill, draw, inner sep = 2pt] (n3) at (1,3){}; \node[] at (1,3.3) {\tiny{$3$}};
			\node[rectangle, rotate=68,  inner sep = 2pt] (xor31) at (-0.5,2.4) {...};
		\node[circle, fill, draw, inner sep = 2pt] (n4) at (1,1.8){}; \node[] at (1.2,2) {\tiny{$4$}};	
	
		\draw[-,color=red] (n3) edge[bend right=40] (xor31); \draw[-latex,color=red] (xor31) edge[bend right=40] (n1);
		\draw[-,color=red] (n2) edge[bend right=40] (xor23); \draw[-latex,color=red] (xor23) edge[bend right=40] (n3); 
		\draw[-,color=red] (n1) edge[bend right=40] (xor12); \draw[-latex,color=red] (xor12) edge[bend right=40] (n2);
	
		\draw[-latex, dashed] (n4) edge[bend right=15] (n3);  
		\draw[-latex, dashed] (n3) edge[bend right=15] (n4);
	
		\draw[-latex, dashed] (n4) edge[bend right=15] (n1);  
		\draw[-latex, dashed] (n1) edge[bend right=15] (n4); 
	
		\draw[-latex, dashed] (n2) edge[bend right=15] (n4); 
		\draw[-latex, dashed] (n4) edge[bend right=15] (n2);
	
		\draw[-latex, dashed] (n1) edge[bend right=15] (n2);
		\draw[-latex, dashed] (n2) edge[bend right=15] (n1);
	\end{tikzpicture}
\end{minipage}
\begin{minipage}[h]{50mm}
	\centering
	$\begin{pmatrix} 
		0 & 1 & 0 & 1\\ 
		1 & 0 & 0 & 1\\ 
		0 & 0 & 0 & 1\\ 
		1 & 1 & 1 & 0
	\end{pmatrix}$
\end{minipage}
\caption{$\mathfrak{G}_1$: The CLAUSE-gadget exists for each clause of the formula. The red lines are not directly part of the gadget and thus are not part of the adjacency matrix.}
\label{clauseGadget}
\end{figure}
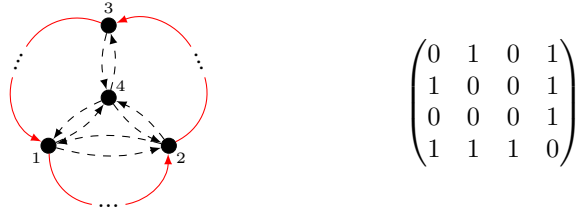

The CLAUSE-gadget has three \textit{external} edges, which are the red lines in Figure \ref{clauseGadget}. These are the edges that connect the clause-gadget to the rest of the graph, to be exact, to the XOR-gadgets. Each of this external edges represents a literal in that clause. Whenever an external edge is part of a cyclic cover, the corresponding literal in that clause is set to $\textsc{FALSE}$. If the edge is not part of the cyclic cover, the literal is set to $\textsc{TRUE}$. Taking all three external edges will leave the middle vertex (node $4$ in Figure \ref{clauseGadget}) unconnected and thus there exists no cyclic cover at all with this partial assignment.

\subsubsection{The VARIABLE-gadget.}
\begin{figure}
\begin{minipage}[h]{\textwidth}
	\centering
	\begin{tikzpicture}
		\tikzstyle{every loop} = [-latex]
		\node[circle, fill, draw, inner sep = 2pt] (LEFT) at (-1,1){}; \node[] at (-1.2,1.2) {\tiny{$1$}};
		\node[circle, fill, draw, inner sep = 2pt] (RIGHT) at (9.5,1){}; \node[] at (9.7,1.2) {\tiny{$n$}};
		\draw[-latex] (RIGHT) to (LEFT); 
		
		\node[circle, fill, draw, inner sep = 2pt] (T1) at (0.5,2){}; \node[] at (0.3,2.2) {\tiny{$2$}};
			\node[rectangle, inner sep = 2pt] (CON1) at (1,2.5) {...};
		\node[circle, fill, draw, inner sep = 2pt] (T2) at (1.5,2){}; \node[] at (1.7,2.2) {\tiny{$3$}};
		\node[circle, fill, draw, inner sep = 2pt] (T3) at (3,2){}; \node[] at (2.8,2.2) {\tiny{$4$}};
			\node[rectangle, inner sep = 2pt] (CON2) at (3.5,2.5) {...};
		\node[circle, fill, draw, inner sep = 2pt] (T4) at (4,2){}; \node[] at (4.2,2.2) {\tiny{$5$}};

		\node[circle, inner sep = 2pt] (DOTS1) at (5.5,2) {...};

		\node[circle, fill, draw, inner sep = 2pt] (T5) at (7,2){}; \node[] at (6.8,2.2) {\tiny{$m$}};
			\node[rectangle, inner sep = 2pt] (CON3) at (7.5,2.5) {...};
		\node[circle, fill, draw, inner sep = 2pt] (T6) at (8,2){}; \node[] at (8.4,2.2) {\tiny{$m+1$}};
		
		\draw[-] (T1) to [out=250, in=290, loop] (T1);\draw[-] (T2) to [out=250, in=290, loop] (T2);
		\draw[-] (T3) to [out=250, in=290, loop] (T3);\draw[-] (T4) to [out=250, in=290, loop] (T4);
		\draw[-] (T5) to [out=250, in=290, loop] (T5);\draw[-] (T6) to [out=250, in=290, loop] (T6);
		\draw[-latex] (T2) to (T3); 

		\draw[-latex] (LEFT) edge[bend left=20] (T1); 
		\draw[-latex] (T6) edge[bend left=20] (RIGHT); 

		\draw[-, color=red] (T1) edge[bend left=20] (CON1); \draw[-latex, color=red] (CON1) edge[bend left=20] (T2);
		\draw[-, color=red] (T3) edge[bend left=20] (CON2); \draw[-latex, color=red] (CON2) edge[bend left=20] (T4);
		\draw[-, color=red] (T5) edge[bend left=20] (CON3); \draw[-latex, color=red] (CON3) edge[bend left=20] (T6);

		\node[circle, fill, draw, inner sep = 2pt] (F1) at (0.5,0){}; \node[] at (0.1,-0.2) {\tiny{$m+2$}};
			\node[rectangle, inner sep = 2pt] (CON4) at (1,-0.5) {...};
		\node[circle, fill, draw, inner sep = 2pt] (F2) at (1.5,0){}; \node[] at (1.9,-0.2) {\tiny{$m+3$}};
			\node[rectangle, inner sep = 2pt] (CON5) at (3.5,-0.5) {...};
		\node[circle, fill, draw, inner sep = 2pt] (F3) at (3,0){}; \node[] at (2.6,-0.2) {\tiny{$m+4$}};
		\node[circle, fill, draw, inner sep = 2pt] (F4) at (4,0){}; \node[] at (4.4,-0.2) {\tiny{$m+5$}};

		\node[circle, inner sep = 2pt] (DOTS2) at (5.5,0) {...};

		\node[circle, fill, draw, inner sep = 2pt] (F5) at (7,0){}; \node[] at (6.6,-0.2) {\tiny{$n-2$}};
			\node[rectangle, inner sep = 2pt] (CON6) at (7.5,-0.5) {...};
		\node[circle, fill, draw, inner sep = 2pt] (F6) at (8,0){}; \node[] at (8.4,-0.2) {\tiny{$n-1$}};

		\draw[-] (F1) to [out=110, in=70, loop] (F1);\draw[-] (F2) to [out=110, in=70, loop] (F2);
		\draw[-] (F3) to [out=110, in=70, loop] (F3);\draw[-] (F4) to [out=110, in=70, loop] (F4);
		\draw[-] (F5) to [out=110, in=70, loop] (F5);\draw[-] (F6) to [out=110, in=70, loop] (F6);
		\draw[-latex] (F2) to (F3); 

		\draw[-latex] (LEFT) edge[bend right=20] (F1); 
		\draw[-latex] (F6) edge[bend right=20] (RIGHT); 

		\draw[-, color=red] (F1) edge[bend right=20] (CON4); \draw[-latex, color=red] (CON4) edge[bend right=20] (F2);
		\draw[-, color=red] (F3) edge[bend right=20] (CON5); \draw[-latex, color=red] (CON5) edge[bend right=20] (F4);
		\draw[-, color=red] (F5) edge[bend right=20] (CON6); \draw[-latex, color=red] (CON6) edge[bend right=20] (F6);
	\end{tikzpicture}
\end{minipage}
\caption{$\mathfrak{G}_2$: The VARIABLE-gadget exists for each distinct variable of the formula. The red lines are not directly part of the gadget.}
\label{variableGadget}
\end{figure}
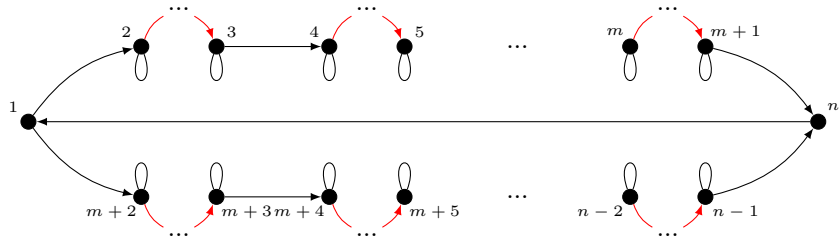

The variable-gadget is constructed for each variable $v_i$ in $\Phi$. For each occurrence of $v_i$ or $\neg v_i$ in $\Phi$, the VARIABLE-gadget has two vertices, whereof the true or the false assignments are on the two different cycles of the gadget. Whenever one of this cycles is used in a cyclic cover the other cyclic is prevented and the nodes must use their self-loop to form a valid cyclic cover. This guarantees, that one literal is either true or false but not both. The shown red lines connect the VARIABLE-gadget to the XOR-gadgets. We write $\mathfrak{G}_1(x_i)$ to address the VARIABLE-gadget for variable $x_i$.

\subsubsection{The XOR-gadget.}
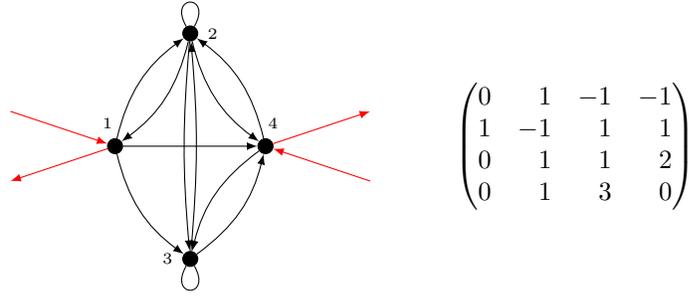
\begin{figure}
\centering
\begin{minipage}[h]{60mm}
	\begin{tikzpicture}
		\tikzstyle{every loop} = [-latex]
		[scale=.8]

		\node[circle,  inner sep = 2pt, thick] (I1) at (0,1) {};
		\node[circle,  inner sep = 2pt, thick] (O2) at (5,1) {};
		\node[circle,  inner sep = 2pt, thick] (I2) at (5,0) {};
		\node[circle,  inner sep = 2pt, thick] (O1) at (0,0) {};
	
		\node[circle, fill, draw, inner sep = 2pt] (n1) at (1.5,0.5){}; \node[] at (1.4,0.8) {\tiny{$1$}};
		\node[circle, fill, draw, inner sep = 2pt] (n3) at (3.5,0.5){}; \node[] at (3.6,0.8) {\tiny{$4$}};
	
		\node[circle, fill, draw, inner sep = 2pt] (n4) at (2.5,2){}; \node[] at (2.8,2) {\tiny{$2$}};
		\draw[-] (n4) to [out=120, in=60, loop] (n4);
		\node[circle, fill, draw, inner sep = 2pt] (n2) at (2.5,-1){}; \node[] at (2.2,-1) {\tiny{$3$}}; 
		\draw[-] (n2) to [out=240, in=300, loop] (n2);
			
		\draw[-latex,color=red] (I1) edge[bend right=0] (n1); \draw[-latex,color=red] (n3) edge[bend right=0] (O2); 
		\draw[-latex,color=red] (n1) edge[bend right=0] (O1); \draw[-latex,color=red] (I2) edge[bend right=0] (n3);
	
		\draw[-latex] (n1) edge[bend left=20] (n4); \draw[-latex] (n4) edge[bend left=20] (n1);  
		\draw[-latex] (n3) edge[bend right=20] (n4); \draw[-latex] (n4) edge[bend right=20] (n3); 
	
		\draw[-latex] (n1) to (n3);
	
		\draw[-latex] (n4) edge[bend right=5] (n2); \draw[-latex] (n2) edge[bend right=5] (n4); 
		\draw[-latex] (n1) edge[bend right=20] (n2);
	
		\draw[-latex] (n2) edge[bend right=20] (n3);
		\draw[-latex] (n3) edge[bend right=20] (n2);
	\end{tikzpicture}
\end{minipage}
\begin{minipage}[h]{50mm}
	$\begin{pmatrix}[r]
		0 & 1 &-1 & -1\\ 
		1 &-1 & 1 & 1\\ 
		0 & 1 & 1 & 2\\ 
		0 & 1 & 3 & 0
	\end{pmatrix}$
\end{minipage}
\caption{The XOR-gadget connects a CLAUSE-gadget with a VARIABLE-gadget. The red lines are not directly part of the gadget.}
\label{XORGadget}
\end{figure}

The XOR-gadget is the most important gadget. It is the only part of the whole final graph, that has edges with weights other than $0$ or $1$. The XOR-gadget connects on the one side to a CLAUSE-gadget and on the other side to a VARIABLE-gadget. As its name suggests, it allows only one side to be taken by a cyclic cover, but not both or not none. The reasoning is that either the VARIABLE-gadget gets the connection, which means that the literal is set to true, or the CLAUSE-gadget gets the connection, which means that that literal in the clause is set to false.

Setting $(i_1,o_1) = (1,4)$ and $(i_2,o_2) = (4,1)$ as the valid input/output ways, Valiant showed, that the adjacency matrix $\mA$ of an XOR-gadget has to fulfill the requirements:
\begin{equation}\label{reqXORgadget}
\begin{matrix}[l]
		\textsc{(A)}\;0 = \perm(\mA) 			& \;\;\; &\textsc{(B)}\;0 = \perm(\mA^{i_1,i_2}_{o_1,o_2})\\ 
		\textsc{(C)}\;c = \perm(\mA^{i_1}_{o_1}) 	& \;\;\; &\textsc{(D)}\;c = \perm(\mA^{i_2}_{o_2})\\
		\textsc{(E)}\;0 = \perm(\mA^{i_1}_{o_2}) 	& \;\;\; &\textsc{(F)}\;0 = \perm(\mA^{i_2}_{o_1}) 
\end{matrix}
\end{equation}
In the matrix he chose for his paper, the constant $c$ is equal to $4$. If we rephrase the requirements (A)-(F) in words, we get: \textit{The cyclic cover can either enter the gadget using node $1$ and leave at node $4$ or the other way round. But taking both ways, none, or taking the wrong exit yields a weight of $0$.}

\textbf{Remark:} For those who are familiar with Valiant's holographic algorithms, one perhaps notices, that these rules are related to the \textit{signature} of a matchgate. The \textit{signature} of a matchgate is enumeration of the weights of the matchgate perfect matchings, that arise if all possible combinations of input and output edges are removed. Here it is the permanent value of the gadget that arise if output and input nodes are removed. 

All three shown gadgets appear as independent gadgets within the graph. All edges that connect between these gadget (the red edges in the figures) have weight $1$, thus do not contribute to the weight of any cyclic-cover.

\subsection{Ben-Dor $\&$ Halevis proof.} In this work, only one gadget is used, which the authors also call CLAUSE-gadget. It represents a clause in the formula, which justifies the name, but actually  has more similarity with the properties of Valiant's XOR-gadget. 

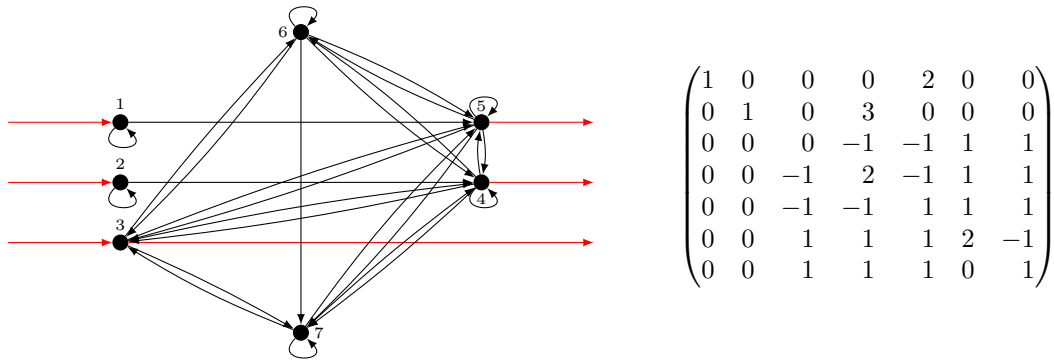
\begin{figure}[h!]
	\begin{minipage}{90mm}
		\begin{tikzpicture}
			[scale=.4]
			\tikzstyle{every loop} = [-latex]
			\node[circle, inner sep = 2pt] (I1) at (0,5)  {}; \node[circle, inner sep = 2pt] (O1) at (20,5)  {};  
			\node[circle, inner sep = 2pt] (I2) at (0,3)  {}; \node[circle, inner sep = 2pt] (O2) at (20,3)  {}; 
			\node[circle, inner sep = 2pt] (I3) at (0,1)  {}; \node[circle, inner sep = 2pt] (O3) at (20,1)  {}; 
			
			\node[circle, fill, draw, inner sep = 2pt] (n1) at (4,5){}; \node[] at (4,5.6) {\tiny{$1$}};
			\node[circle, fill, draw, inner sep = 2pt] (n2) at (4,3){}; \node[] at (4,3.6) {\tiny{$2$}};
			\node[circle, fill, draw, inner sep = 2pt] (n3) at (4,1){}; \node[] at (4,1.6) {\tiny{$3$}};
				
			\node[circle, fill, draw, inner sep = 2pt] (n6) at (10,8){}; \node[] at (9.4,8) {\tiny{$6$}};
			\node[circle, fill, draw, inner sep = 2pt] (n7) at (10,-2){}; \node[] at (10.6,-2) {\tiny{$7$}};
			
			\node[circle, fill, draw, inner sep = 2pt] (n4) at (16,3){}; \node[] at (16,2.45) {\tiny{$4$}};
			\node[circle, fill, draw, inner sep = 2pt] (n5) at (16,5){}; \node[] at (16,5.55) {\tiny{$5$}};
			
			\draw[-latex,color=red] (I1) to (n1); \draw[-latex,color=red] (n5) to (O1);
			\draw[-latex,color=red] (I2) to (n2); \draw[-latex,color=red] (n4) to (O2);
			\draw[-latex,color=red] (I3) to (n3); \draw[-latex,color=red] (n3) to (O3);
			
			\draw[-latex] (n1) edge node[auto] {} (n5);
	
			\draw[-latex] (n2) edge node[auto] {} (n4);
			
			\draw[-latex] (n3) edge[bend left=3] node[auto] {} (n4);
			\draw[-latex] (n3) edge[bend left=3] node[auto] {} (n5);
			\draw[-latex] (n3) edge[bend left=3] node[auto] {} (n6);
			\draw[-latex] (n3) edge[bend left=3] node[auto] {} (n7);
	
			\draw[-latex] (n4) edge[bend left=3] node[auto] {} (n3);
			\draw[-latex] (n4) edge[bend left=10] node[auto] {} (n5);
			\draw[-latex] (n4) edge[bend right=5] node[auto] {} (n6);
			\draw[-latex] (n4) edge[bend left=3] node[auto] {} (n7);
	
			\draw[-latex] (n5) edge[bend left=3] node[auto] {} (n3);
			\draw[-latex] (n5) edge[bend left=10] node[auto] {} (n4);
			\draw[-latex] (n5) edge[bend left=5] node[auto] {} (n6);
			\draw[-latex] (n5) edge[bend left=5] node[auto] {} (n7);
	
			\draw[-latex] (n6) edge[bend left=5] node[auto] {} (n3);
			\draw[-latex] (n6) edge[bend right=5] node[auto] {} (n4);
			\draw[-latex] (n6) edge[bend left=5] node[auto] {} (n5);
			\draw[-latex] (n6) edge node[auto] {} (n7);
	
			\draw[-latex] (n7) edge[bend left=5] node[auto] {} (n3);
			\draw[-latex] (n7) edge[bend left=3] node[auto] {} (n4);
			\draw[-latex] (n7) edge[bend left=5] node[auto] {} (n5);
	
			
			\draw[-latex] (n1) to [in=315,out=225, loop]  node[auto] {} (n1);
			\draw[-latex] (n2) to [in=315,out=225, loop]  node[auto] {} (n2);
			\draw[-latex] (n4) to [in=315,out=225, loop]  node[auto] {} (n4);
			\draw[-latex] (n5) to [in=45,out=135, loop]  node[auto] {} (n5);
			\draw[-latex] (n6) to [in=45,out=135, loop]  node[auto] {} (n6);
			\draw[-latex] (n7) to [in=315,out=225, loop]  node[auto] {} (n7);	
		\end{tikzpicture}
	\end{minipage}
	\begin{minipage}[h]{40mm}
		$\begin{pmatrix}[r]
			1 & 0 & 0 & 0 & 2 & 0 & 0\\ 
			0 & 1 & 0 & 3 & 0 & 0 & 0\\ 
			0 & 0 & 0 &-1 &-1 & 1 & 1\\ 
			0 & 0 &-1 & 2 &-1 & 1 & 1\\ 
			0 & 0 &-1 &-1 & 1 & 1 & 1\\ 
			0 & 0 & 1 & 1 & 1 & 2 &-1\\ 
			0 & 0 & 1 & 1 & 1 & 0 & 1
		\end{pmatrix}$
	\end{minipage}
\caption{The CLAUSE-gadget of Ben-Dor and Halevi. The red lines are not directly part of the gadget.}
\label{clausegadget2}
\end{figure}
They stated the following rules in order to make their proof work.
\begin{equation}\label{reqCLAUSEgadget}
\begin{matrix}[l]
	\textsc{(A)}\;0 = \perm(\mA) 			& \;\;\; &\textsc{(B)}\;0 = \perm(\mA^{i_1}_{o_2})		& \;\;\; &	\textsc{(C)}\;0 = \perm(\mA^{i_1}_{o_3})\\	
	\textsc{(D)}\;0 = \perm(\mA^{i_2}_{o_1})	& \;\;\; &\textsc{(E)}\;0 = \perm(\mA^{i_2}_{o_3})		& \;\;\; &	\textsc{(F)}\;0 = \perm(\mA^{i_3}_{o_1})\\		
	\textsc{(G)}\;0 = \perm(\mA^{i_3}_{o_2})	& \;\;\; &\textsc{(H)}\;0 = \perm(\mA^{i_1,i_2}_{o_1,o_3})	& \;\;\; &	\textsc{(G)}\;0 = \perm(\mA^{i_1,i_2}_{o_2,o_3})\\
	\textsc{(I)}\;0 = \perm(\mA^{i_1,i_3}_{o_1,o_2})& \;\;\; &\textsc{(J)}\;0 = \perm(\mA^{i_1,i_3}_{o_2,o_3})	& \;\;\; &	\textsc{(K)}\;0 = \perm(\mA^{i_2,i_3}_{o_1,o_3})\\
	\textsc{(L)}\;0 = \perm(\mA^{i_2,i_3}_{o_1,o_2})& \;\;\; &\textsc{(M)}\;c = \perm(\mA^{i_1,i_2,i_3}_{o_1,o_2,o_3})& \;\;\; &	\textsc{(N)}\;c = \perm(\mA^{i_1,i_2}_{o_1,o_2})\\
	\textsc{(O)}\;c = \perm(\mA^{i_1,i_3}_{o_1,o_3})& \;\;\; &\textsc{(P)}\;c = \perm(\mA^{i_2,i_3}_{o_2,o_3})	& \;\;\; &	\textsc{(Q)}\;c = \perm(\mA^{i_1}_{o_1})\\
	\textsc{(R)}\;c = \perm(\mA^{i_2}_{o_2})& \;\;\; &\textsc{(S)}\;c = \perm(\mA^{i_3}_{o_3})	& \;\;\; &	\\
\end{matrix}
\end{equation}
In their paper they chose $(i_1,i_2,i_3)=(1,2,3)$ and $(o_1,o_2,o_3)=(3,4,5)$. The rules can be subsumed shortly by: \textit{Each CLAUSE-gadget must be used and no input-output switching is allowed (i.e., if one enters at $i_j$ one must leave at $o_j$), since this yields a permanent of zero.}

\subsection{A preventing Identity.}
In the computational complexity lecture notes of Oded Goldreich one can find the following lines regarding Valiant's proof
\begin{quote}\textit{
 ``If there a transformation from $\Phi$ to $\grG$, such that $\perm(\mA) = c^m\#\Phi$, then assuming $\classNP \varsubsetneq \classBPP$, the constant $c$ must be even.''
}
\end{quote}
This statement refers to the constant $c$ in the rules for the XOR-gadget (see Rules \ref{reqXORgadget}) and which is $c=4$ in Valiant's proof. If the constant would be $c=1$ one has a parsimonious reduction to the permanent, but this would contradict the assumption $\classNP \varsubsetneq \classBPP$. In the next lines, we show that this would not only contradict that well believed complexity assumption, but also a known determinant identity. The tool we use is that we can evaluate the permanent modulo $2$ since the determinant and the permanent are equal in $\mathbb{Z}/2\mathbb{Z}$. 

One could observe that there is no need to restrict the requirements of the XOR-gadget to such narrow rules like the ones described in Rules \ref{reqXORgadget}. If one is only interested in deciding whether the formula has a satisfying assignment, one bit of information is enough and thus computing modulo $2$ is sufficient. This weakening can be translated to the following new rules for the XOR-gadget:
\begin{equation}\label{reqXORgadget2}
\begin{matrix}[l]
	\textsc{(A)}\;0 \equiv \perm(\mA) \equiv \detm(\mA) \pmod{2} 		& \;\;\; &\textsc{(B)}\;0 \equiv \perm(\mA^{i_1,i_2}_{o_1,o_2}) \equiv \detm(\mA^{i_1,i_2}_{o_1,o_2}) \pmod{2}\\ 
	\textsc{(C)}\;1 \equiv \perm(\mA^{i_1}_{o_1}) \equiv \detm(\mA^{i_1}_{o_1}) \pmod{2} 	& \;\;\; &\textsc{(D)}\;1 \equiv \perm(\mA^{i_2}_{o_2}) \equiv \detm(\mA^{i_2}_{o_2}) \pmod{2} \\
	\textsc{(E)}\;0 \equiv \perm(\mA^{i_1}_{o_2}) \equiv \detm(\mA^{i_1}_{o_2}) \pmod{2} 	& \;\;\; &\textsc{(F)}\;0 \equiv \perm(\mA^{i_2}_{o_1}) \equiv \detm(\mA^{i_2}_{o_1}) \pmod{2}
\end{matrix}
\end{equation}
Suppose, we are given a boolean formula from $\udreisat$. All cyclic covers that do not yield a satisfying assignment sum up to an even weight and thus vanish modulo $2$. If there is a satisfying assignment, then this satisfying assignment yields an odd sum, which could be detected by evaluating the determinant modulo $2$. 

However, regarding the determinant there exists an identity that prevents someone from finding such a gadget.

\begin{theorem}[Desnanot-Jacobi identity ($\dji$)]
 Let $\mM = m_{i,j}$, $1 \leq i,j \leq n$ be a square matrix 
\begin{equation}
 \det(\mM)\det(\mM^{i_1,i_2}_{o_1,o_2}) = \det(\mM^{i_1}_{o_1})\det(\mM^{i_2}_{o_2}) - \det(\mM^{i_2}_{o_1})\det(\mM^{i_1}_{o_2})
\end{equation}
\end{theorem}
One finds the Desnanot-Jacobi identity mostly in the literature in connection with Dodgson condensation algorithm \cite{Dog66}. If one applies the $\dji$ to the weakened Rules \ref{reqXORgadget2} for the XOR-gadget, one gets
\begin{equation*}
 \detm(\mA)\detm(\mA^{i_1,i_2}_{o_1,o_2}) = \detm(\mA^{i_1}_{o_1})\detm(\mA^{i_2}_{o_2}) - \detm(\mA^{i_1}_{o_2})\detm(\mA^{i_2}_{o_1}) \Leftrightarrow 0\cdot 0 \equiv 1\cdot 1 - 0\cdot 0\pmod{2}
\end{equation*}
which is false. Similar, if one applies this to the rules of of Ben-Dor and Halevis CLAUSE-gadget, one gets
\begin{equation*}
 \detm(\mA)\detm(\mA^{i_1,i_2}_{o_1,o_2}) = \detm(\mA^{i_1}_{o_1})\detm(\mA^{i_2}_{o_2}) - \detm(\mA^{i_1}_{o_2})\detm(\mA^{i_2}_{o_1}) \Leftrightarrow 0\cdot 1 \equiv 1\cdot 1 - 0\cdot 0\pmod{2}
\end{equation*}
which is also false. Additionally, if one lays down the rules for the EQUALITY-gadget of Blaeser et al. \cite{Bla07} one gets
\begin{equation*}
 \detm(\mA)\detm(\mA^{i_1,i_2}_{o_1,o_2}) = \detm(\mA^{i_1}_{o_1})\detm(\mA^{i_2}_{o_2}) - \detm(\mA^{i_1}_{o_2})\detm(\mA^{i_2}_{o_1}) \Leftrightarrow 1\cdot 1 \equiv 0\cdot 0 - 0\cdot 0\pmod{2}
\end{equation*}
which again is false\footnote{The EQUALITY-gadget could actually easily be used to replace the XOR-gadget}. So, the ability to decide if a formula from $\udreisat$ has a satisfying assignment by trying to find a suitable gadget can not succeed.  

\begin{proposition}
The constant $c$ in Valiants proof can not be odd due to the Desnanot-Jacobi identity.
\end{proposition}

\textbf{A $\dji$-compatible gadget.} One could ask the question: What kind of useful gadget is possible, that is compatible with the $\dji$? The answer is: A PLANARITY-gadget. A PLANARITY-gadget substitutes each
edge crossing with a planar subgraph and can be described with the rules
\begin{equation}\label{reqOpenQuestion}
\begin{matrix}[l]
	\textsc{(A)}\;1 \equiv \perm(\mA) \equiv \detm(\mA) \pmod{2} 		& \;\;\; &\textsc{(B)}\;1 \equiv \perm(\mA^{i_1,i_2}_{o_1,o_2}) \equiv \detm(\mA^{i_1,i_2}_{o_1,o_2}) \pmod{2}\\ 
	\textsc{(C)}\;1 \equiv \perm(\mA^{i_1}_{o_2}) \equiv \detm(\mA^{i_1}_{o_2}) \pmod{2} 	& \;\;\; &\textsc{(D)}\;1 \equiv \perm(\mA^{i_2}_{o_1}) \equiv \detm(\mA^{i_2}_{o_1}) \pmod{2}\\
	\textsc{(E)}\;0 \equiv \perm(\mA^{i_1}_{o_1}) \equiv \detm(\mA^{i_1}_{o_1}) \pmod{2} 	& \;\;\; &\textsc{(F)}\;0 \equiv \perm(\mA^{i_2}_{o_2}) \equiv \detm(\mA^{i_2}_{o_2}) \pmod{2}
\end{matrix}
\end{equation}
, which translates into
\begin{equation*}
 \detm(\mA)\detm(\mA^{i_1,i_2}_{o_1,o_2}) = \detm(\mA^{i_1}_{o_1})\detm(\mA^{i_2}_{o_2}) - \detm(\mA^{i_1}_{o_2})\detm(\mA^{i_2}_{o_1}) \Leftrightarrow 1\cdot 1 \equiv 0\cdot 0 - 1\cdot 1\pmod{2}
\end{equation*}
and is true. 

In the following we introduce the following notation regarding the rules of a gadget.
\begin{defin}[Signature] For each gadget $\gad$ with the input-output pairs $(i_1,o_1)$ and $(i_2,o_2)$ we define the signature $\mathcal{S}_\gad$ as
 \begin{equation*}
  \mathcal{S}_\gad \define \left(r_1, r_2, r_3, r_4, r_5, r_6 \right)_p
 \end{equation*}
  whereof the $r_i$'s are from
 \begin{equation*}
	\begin{matrix}[l]
		\textsc{(A)}\;\perm(\mA) \equiv r_1 \pmod{p} & \;\;\; &\textsc{(B)}\;\perm(\mA^{i_1,i_2}_{o_1,o_2}) \equiv r_2 \pmod{p}\\
		\textsc{(C)}\;\perm(\mA^{i_1}_{o_1}) \equiv r_3 \pmod{p} & \;\;\; &\textsc{(D)}\;\perm(\mA^{i_2}_{o_2}) \equiv r_4 \pmod{p}\\
		\textsc{(E)}\;\perm(\mA^{i_1}_{o_2}) \equiv r_4 \pmod{p} & \;\;\; &\textsc{(F)}\;\perm(\mA^{i_2}_{o_1}) \equiv r_5 \pmod{p}
	\end{matrix}
 \end{equation*}
 When the signature array has no subscript $p$ then we refer to the integer values rather than the modulo $p$ reduced value of the permanent. $\blacktriangleleft$
\end{defin}

\section{A planar Valiant graph.}

Our progress depends on the observation that the $\dji$ is a determinant identity and a similar identity for the permanent is not known. The CLAUSE-gadget as well as the VARIABLE-gadget are uncritical, since they already are unweighted graphs. But the XOR-gadget has edges with weights different from $0$ and $1$, which does not allow to apply Proposition \ref{propPermPerf}. For a potential
new XOR-gadget we searched for $(0/1)$-matrices that has to following properties over $\mathbb{Z}/p\mathbb{Z}$:
\begin{itemize}
 \item whenever an disallowed path is taken, i.e., one that can not lead to a satisfying assignment, the permanent vanishes in $\mathbb{Z}/p\mathbb{Z}$.
 \item whenever an allowed path is taken, i.e., one that can lead to a satisfying assignment the permanent does not vanish in $\mathbb{Z}/p\mathbb{Z}$.
\end{itemize}
The weakening takes place in the second point. We do not force a single constant $c$, but only a value that is different from zero. That means the signature of a potential XOR-gadget is of the form
\begin{equation} 
 \mathcal{S} = (0,0,\neq 0,\neq 0,0,0)_p
\end{equation}

The new XOR-gadget $\mathfrak{G}_3$ we found, for $p=3$, can be seen in Figure \ref{newXORgadget}. 

\begin{figure}
\centering
\begin{minipage}[h]{80mm}
	\begin{tikzpicture}
 		[scale=.8]
		\tikzstyle{every loop} = [-latex]
		\node[circle,  inner sep = 2pt, thick] (I1) at (2,0) {};
		\node[circle,  inner sep = 2pt, thick] (I2) at (2,3) {};
		\node[circle,  inner sep = 2pt, thick] (O1) at (9,3) {};
		\node[circle,  inner sep = 2pt, thick] (O2) at (9,0) {};

		\node[circle, fill, draw, inner sep = 2pt] (n1) at (4,0){}; \node[] at (4,-0.3) {\tiny{$1$}};
		\node[circle, fill, draw, inner sep = 2pt] (n2) at (4,3){}; \node[] at (4,3.3) {\tiny{$2$}};
		\node[circle, fill, draw, inner sep = 2pt] (n3) at (7,3){}; \node[] at (7.2,2.8) {\tiny{$3$}};
		\node[circle, fill, draw, inner sep = 2pt] (n4) at (7,0){}; \node[] at (7.2,0.2) {\tiny{$4$}};
		\node[circle, fill, draw, inner sep = 2pt] (n5) at (5.5,0.5){}; \node[] at (5.8,0.5) {\tiny{$5$}};
		\node[circle, fill, draw, inner sep = 2pt] (n6) at (5.5,2.5){}; \node[] at (5.8,2.4) {\tiny{$6$}};
		
		\draw[-latex] (n1) to (n3);
		\draw[-latex] (n1) to (n6);
		
		\draw[-latex] (n2) to (n1);
		\draw[-latex] (n2) to (n4);
		\draw[-latex] (n2) edge[bend left=10] (n6);

		\draw[-latex] (n3) to (n4);
		\draw[-latex] (n3) to [in=60,out=120, loop] (n3);
		\draw[-latex] (n3) to (n6);
		
		\draw[-latex] (n4) to [in=240,out=300, loop] (n4);
		\draw[-latex] (n4) edge[bend left=10] (n6);

		\draw[-latex] (n5) to (n2);
		\draw[-latex] (n5) to (n3);
		\draw[-latex] (n5) to [in=240,out=300, loop] (n5);

		\draw[-latex] (n6) edge[bend left=10] (n2);
		\draw[-latex] (n6) edge[bend left=10] (n4);
		\draw[-latex] (n6) to (n5);
		\draw[-latex] (n6) to [in=60,out=120, loop] (n6);

		\draw[-latex, color=red] (I1) edge[bend left=0] (n1);		
		\draw[-latex, color=red] (I2) edge[bend left=0] (n2);
		\draw[-latex, color=red] (n3) edge[bend left=0] (O1);
		\draw[-latex, color=red] (n4) edge[bend left=0] (O2);
	\end{tikzpicture}
\end{minipage}
\begin{minipage}[h]{50mm}
$\begin{pmatrix}[r]
 	0 & 0 & 1 & 0 & 0 & 1\\
	1 & 0 & 0 & 1 & 0 & 1\\
	0 & 0 & 1 & 1 & 0 & 1\\
	0 & 0 & 0 & 1 & 0 & 1\\
	0 & 1 & 1 & 0 & 1 & 0\\
	0 & 1 & 0 & 1 & 1 & 1
\end{pmatrix}$
\end{minipage}
\caption{$\mathfrak{G}_{\text{3}}$: The new XOR-gadget. The red lines are not directly part of the gadget.}
\label{newXORgadget}
\end{figure}
It has $1$,$2$ as input nodes and $3$,$4$ as its output nodes. The \textit{correct}\footnote{That means, the path that could lead to a satisfying assignment} path is from $1$ to $3$ and $2$ to $4$. Entering in $1$ and leaving in $4$ (or $2$ and $3$) yields, according to its signature, a permanent of $0$ in $\mathbb{Z}/3\mathbb{Z}$. Its signature is
\begin{equation}\label{reqnewXORgadget}
 \mathcal{S}_{\gad_3} = (6,3,4,4,6,3) = (0,0,1,1,0,0)_3
\end{equation}

We now proof the first theorem, which shows that already this new gadget allows to efficiently decide the satisfiability of a formula from $\udreisat$ if $\hat{\grG}$ turns out to be planar. 

\begin{theorem}\label{theorem2a}
 Let $\Phi$ be boolean formula from $\udreisat$ with $n$ variables and $m$ clauses and let $\grG$ be the graph that is constructed from $\Phi$ according to Valiant's proof, using the gadgets $\mathfrak{G}_1$ (see Figure \ref{clauseGadget}), $\mathfrak{G}_2$ (see Figure \ref{variableGadget}), $\mathfrak{G}_3$ (see Figure \ref{newXORgadget}). If $\hat{\grG}$ is planar, then there exists an algorithm $\mathcal{A}$ that decides $\Phi$'s satisfiability in polynomial time $\mathcal{O}(m^3)$.
\end{theorem}

\begin{proof}
	For the proof, we follow the arguments of Valiant in a close way, but use the notation of an $F$-completion from Ben-Dor and Halevi (see \cite{Ben93}, Definition 8). We do not rewrite the whole proof, but only argue along the essential step, that is, where the sum over all cyclic-covers is computed. We describe an algorithm that does:
	\begin{alltt}
	\small{\(\mathcal{A}(\Phi)\):
	\(\grG \leftarrow\) buildValiantGraph(\(\gad\sb{1},\gad\sb{2},\gad\sb{3},\Phi\))
	\(\mathsf{s} \leftarrow \mathsf{PerfMatch}(\hat{\grG})\pmod{3}\)
	return \(\mathsf{1}\sb{\mathsf{s}>0}\)}
	\end{alltt}
	The input of $\mathcal{A}$ is a formula $\Phi$ from $\udreisat$. The only gadget we have substituted so far is the XOR-gadget. In all cases where a cyclic cover in Valiant's proof contributes a weight of $0$, the same cyclic cover in our case contributes a weight that is $0$ or a multiple of $3$. And in all cases where in Valiant's proof the cyclic cover contributes a weight of $4^{3m}$, the same cyclic cover in our case contributes a weight that is congruent to $1$ modulo $3$. Since any cyclic cover must be of one of these types, the final permanent value is the sum over
\begin{align*}
\perm(\mA) & = \sum_{C \in \mathcal{C}^F} W(C) = 3\cdot \mathbb{Z} + \mathsf{1}_{\phi\text{ is sat}}(1+3\cdot \mathbb{Z})^{3m} = 3\cdot \mathbb{Z} + \mathsf{1}_{\phi\text{ is sat}}(1 + 3\cdot \mathbb{Z})\\
	   & = \mathsf{PerfMatch}(\hat{\grG})
\end{align*}
After reducing modulo $3$ we get
\begin{equation}
	\mathsf{PerfMatch}(\hat{\grG}) \pmod{3} =  \mathsf{1}_{\phi\text{ is sat}} = \begin{cases}
	              0 & \Phi \text{ is unsatisfiable}\\
		      1 & \Phi \text{ is satisfiable}\\
	             \end{cases}
\end{equation}

If $\Phi$ consists of $n$ variables $x_1,...,x_n$ and $m$ clauses. $x_i$ together with $\neg x_i$ occurs $n_i$ times, thus $\sum^n_{i=1} n_i \leq 3m$ as well as $n \leq 3m$, since $\Phi$ is a 3CNF.
Hence, the graph $\grG$ consists of 
	\begin{align*}
	m|\mathfrak{G}_1| + \sum^n_{i=1}|\mathfrak{G}_2(x_i)| +  3m|\mathfrak{G}_3| 	& = m\cdot 4 + \sum^n_{i=1}(2+2n_i) + 3m\cdot 6 \\
											& = 2n+2\sum^n_{i=1}n_i + 22m \\
											& < 6m + 6m + 22m \\
											& < 34m 
	\end{align*}
	nodes. The bipartite double cover operations doubles the number of nodes, so we can bound the number of nodes from above via $|\hat{\grG}| \leq 68m$. Finally, Kasteleyn's algorithms runs in time 
	that is cubic in the number of nodes, which yields a total running time of our algorithm of $\mathcal{O}(68^3m^3)$, or $\mathcal{O}(m^3)$ for large $m$.
\end{proof}

The following proposition shows the well known consequence between deciding satisfiability and finding a solution.

\begin{corollary}
 If one could decide the satisfiability of a formula $\Phi$ from $\udreisat$ that has $n$ distinct variables in time $T$, one can find the solution in time $\Omega(n)T$.
\end{corollary}

\begin{proof}
 To find the solution, if it exists, apply to following algorithm, which makes calls to the oracle $O^{\udreisat}(\Phi)$. That oracle implements the algorithm from Theorem $\ref{theorem2a}$ and returns $0$ or $1$ according to $\Phi$'s satisfiability:\\

\begin{minipage}[h]{100mm}
  1. $v \leftarrow O^{\udreisat}(\Phi)$.\\
  2. assignment $\leftarrow$ [$0$, ..., $0$] \\
  3. if $v = 0$ then return \texttt{'unsatisfiable'}.\\
  4. for $i = 1$ to $n$ do\\
  5. $\;\;$ $v \leftarrow O^{\udreisat}(\Phi_{x_i = 1})$.\\
  6. $\;\;$ if $v = 0$ then assignment[i] = $0$\\
  7. $\;\;$ else assignment[i] = $1$\\
  8. return assignment\\
\end{minipage}

 At the end, in the array \textit{assignment} at position $i$, a value of the $i$-th variable is stored that leads to a satisfying assignment. The loop makes $n$ calls to the oracle, and each loop is evaluated in time $T$.
\end{proof}

\begin{theorem}\label{theorem2b}
 Let $\Phi$ be random boolean formula from $\dreisat$ with $n$ variables and $m$ clauses and let $\grG$ be the graph that is constructed from $\Phi$ according to Valiant's proof, using the gadgets $\mathfrak{G}_1$ (see Figure \ref{clauseGadget}), $\mathfrak{G}_2$ (see Figure \ref{variableGadget}), $\mathfrak{G}_3$ (see Figure \ref{newXORgadget}). If $\hat{\grG}$ is planar, then there exists an algorithm $\mathcal{A}$ that decides $\Phi$'s satisfiability in randomized polynomial time $\mathcal{O}(km^3)$ with an a success probability of 
\begin{equation*}
	\mathsf{Pr}[\mathsf{s} \leftarrow \mathcal{A}(\Phi,k) | \mathsf{1}_{\Phi\;\mathrm{is\;sat}} = \mathsf{s} ] = 1 - \mathsf{s}\frac{1}{3^k} 
\end{equation*}
\end{theorem}

\begin{proof}
 The proof is similar to the one of Theorem \ref{theorem2a} so we can overtake the same algorithm $\mathcal{A}$ and its running time. For the proof, one could also use function $f$ from Valiant-Vazarani and convert the formula into a formula from $\udreisat$. But that will increase unnecessary the number of nodes in the resulting Valiant graph.

The difference here is, that for a formula from $\dreisat$ it could happen, that the number of satisfying solutions is a multiple of $3$, which can not be distinguished from the case that the formula is unsatisfiable. Assume the following algorithm $\mathcal{B}$:\\

\begin{minipage}[h]{80mm}
1. for $i = 1$ to $n$ do \\
2. $\;\;\mathsf{s} \leftarrow \mathcal{A}(\Phi_{x_i = 1})$\\
3. $\;\;$if $\mathsf{s} = 1$ then return \texttt{'satisfiable'}\\
4. return \texttt{'unsatisfiable'}\\
\end{minipage}

Case 1: $\Phi$ is unsatisfiable. The algorithm will never return \texttt{'satisfiable'} in Step 3, since no substitution of $x_i=1$ can make the formula satisfiable. So the algorithm always returns the correct result.

Case 2: $\Phi$ is satisfiable. If the algorithm returns \texttt{'satisfiable'}, the formula must be satisfiable and we are done. Whenever the algorithm returns \texttt{'unsatisfiable'}, each solution set of the formula with $x_i = 1$ must again be a multiple of $3$. Since we have a random boolean formula, there is no dependency among the bits of the solutions, the chance that this happens for each $\Phi_{x_i=1}$ is $3^{-1}$. Thus, the chance after $k$-trials not to detect the satisfiability is $3^{-k}$.
\end{proof}

\section{Enhancements to reach planarity of $\hat{\grG}$.} 

To enhance the previous approach, we try to make the Valiant graph already planar. The reason is that for large graphs with a low connectivity it turns out that already planar graphs do more often have planar bipartite double covers than non-planar graphs.\footnote{Unfortunately, we can base this statement only on a few calculation we did.}
In order to achieve this, we restrict our input to $\textit{pn-planar}$ formulas. 
\begin{figure}[h!]
	\centering
	\begin{minipage}[h]{60mm}
		\begin{tikzpicture}
			[scale=.8]

			\node[circle, draw, inner sep = 2pt] (n1) at (1,0) {1};
			\node[circle, draw, inner sep = 2pt] (n2) at (2,0) {2};
			\node[circle, draw, inner sep = 2pt] (n3) at (3,0) {3};
			\node[circle, draw, inner sep = 2pt] (n4) at (4,0) {4};
			\node[circle, draw, inner sep = 2pt] (n5) at (5,0) {5};
			\node[circle, draw, inner sep = 2pt] (n6) at (6,0) {6};
	
			\node[circle, draw, inner sep = 2pt] (c1) at (1,2) {c1};
			\node[circle, draw, inner sep = 2pt] (c2) at (2.6,2) {c2};
			\node[circle, draw, inner sep = 2pt] (c3) at (4.4,2) {c3};
			\node[circle, draw, inner sep = 2pt] (c4) at (6,2) {c4};
			
			\draw[-] (n1) to (c1);\draw[-] (n1) to (c2);\draw[-] (n1) to (c3);
			\draw[-] (n2) to (c1);
			\draw[-] (n3) to (c1);\draw[-] (n3) to (c3);
			\draw[-] (n4) to (c2);\draw[-] (n4) to (c4);
			\draw[-] (n5) to (c3);\draw[-] (n5) to (c4);
			\draw[-] (n6) to (c2);\draw[-] (n6) to (c4);
			
		\end{tikzpicture}
	\end{minipage}
	\begin{minipage}[h]{50mm}
		\begin{tikzpicture}
			[scale=.8]
			\node[circle, draw, inner sep = 2pt] (n1) at (1,0) {1};
			\node[circle, draw, inner sep = 2pt] (n2) at (2,0) {2};
			\node[circle, draw, inner sep = 2pt] (n3) at (3,0) {3};
			\node[circle, draw, inner sep = 2pt] (n4) at (4,0) {4};
			\node[circle, draw, inner sep = 2pt] (n5) at (5,0) {5};
			\node[circle, draw, inner sep = 2pt] (n6) at (6,0) {6};
			
			\node[circle, draw, inner sep = 2pt] (c1) at (1,2) {c1};
			\node[circle, draw, inner sep = 2pt] (c2) at (2.6,2) {c2};
			\node[circle, draw, inner sep = 2pt] (c3) at (4.4,-2) {c3};
			\node[circle, draw, inner sep = 2pt] (c4) at (6,2) {c4};

			\node[circle, inner sep = 0pt] (h1) at (0,2.5) {};
			\node[circle, inner sep = 0pt] (h2) at (7,2.5) {};
			
			\draw[-] (n1) to (c1);\draw[-] (n1) to [bend left=50] (h1);\draw[-] (h1) to [bend left=50] (c2);\draw[-] (n1) to (c3);
			\draw[-] (n2) to (c1);
			\draw[-] (n3) to (c1);\draw[-] (n3) to (c3);
			\draw[-] (n4) to (c2);\draw[-] (n4) to (c4);
			\draw[-] (n5) to (c3);\draw[-] (n5) to (c4);
			\draw[-] (n6) to [bend right=50] (h2);\draw[-] (h2) to [bend right=50] (c2);\draw[-] (n6) to (c4);

			\node[rectangle, draw, minimum width = 180pt, minimum height = 90pt, rounded corners=2pt, dashed] (face1) at (3.5,2) {};
			\node[rectangle, minimum width = 50pt, minimum height = 20pt, rounded corners=2pt] (face1) at (0.5,3.5) {Face 1};
			\node[rectangle, minimum width = 50pt, minimum height = 20pt, rounded corners=2pt] (face2) at (0.5,-0.7) {Face 2};
		\end{tikzpicture}
	\end{minipage}
\label{twoDrawings}
\caption{Two different drawings of the incidence graph of the formula $\Phi = (x_1 \vee x_2 \vee x_3) \wedge (x_1 \vee \neg x_4 \vee x_6) \wedge  (\neg x_1 \vee \neg x_3 \vee x_5) \wedge (\neg x_4 \vee \neg x_5 \vee x_6)$. On the left in a typical bipartite form and on the right the $\textit{pn-planar}$ embedding (without the edges that connect the variable-nodes).}
\end{figure}
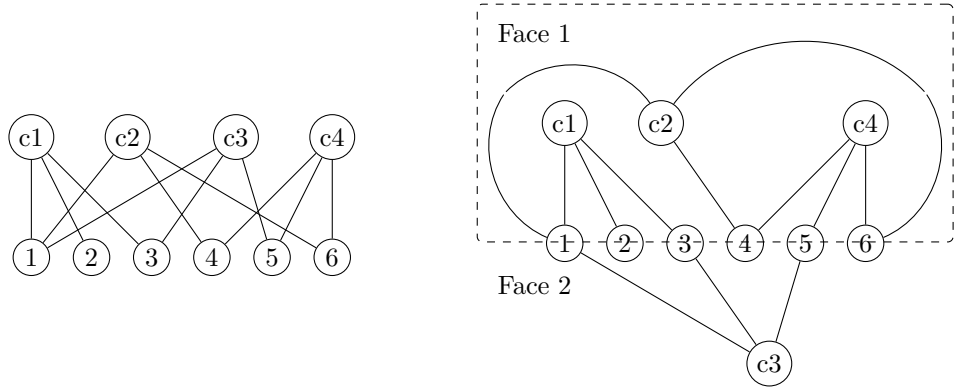

In Figure \ref{twoDrawings} an example of a formula with an $\textit{pn-planar}$ embedding is shown. The reason to focus on this kind of formulas is, that they can exactly be retraced with the graph gadgets from Valiant. E.g., the VARIABLE-gadget has two directions. One for the positive literals and one for the negative literals. By the usage of pn-planar formulas, it is possible to turn each gadget into the correct face direction to keep planarity. 

To make the entire graph planar, all gadgets have to be planar itself. However, the current new XOR-gadget $\mathfrak{G}_3$ is not planar. Even more, simple planarity is not sufficient. The gadgets must be \textbf{circular planar} with boundaries that consider their input and output nodes as well as the face on which a clause is located. We will describe these said requirements in the next sections in more detail.

\subsubsection{The RL-XOR-gadget} The first gadget we found can be seen in Figure \ref{RLXORgadget}. 

\begin{figure}[h!]
\centering
\begin{minipage}[h]{80mm}
	\begin{tikzpicture}
		[scale=.8]
		\tikzstyle{every loop} = [-latex]
		\node[circle, draw, inner sep = 36pt, dashed] (circ) at (3,2.1) {};

		\node[circle,  inner sep = 2pt, thick] (O1) at (-1,1) {};
		\node[circle,  inner sep = 2pt, thick] (O2) at (7,3) {};
		\node[circle,  inner sep = 2pt, thick] (I2) at (-1,3) {};
		\node[circle,  inner sep = 2pt, thick] (I1) at (7,1) {};
	
		\node[circle, fill, draw, inner sep = 2pt] (n1) at (5,1){}; \node[] at (5.2,0.8) {\tiny{$1$}};
		\node[circle, fill, draw, inner sep = 2pt] (n2) at (1,3){}; \node[] at (0.7,2.8) {\tiny{$2$}};
		\node[circle, fill, draw, inner sep = 2pt] (n3) at (1,1){}; \node[] at (0.7,0.8) {\tiny{$3$}};  
		\node[circle, fill, draw, inner sep = 2pt] (n4) at (5,3){}; \node[] at (5.2,3.3) {\tiny{$4$}};
		\node[circle, fill, draw, inner sep = 2pt] (n5) at (3,3.8){}; \node[] at (3,3.5) {\tiny{$5$}};
		\node[circle, fill, draw, inner sep = 2pt] (n6) at (2,2){}; \node[] at (2,2.3) {\tiny{$6$}};
		
		\draw[-latex] (n1) edge[bend left=0] (n3);
		\draw[-latex] (n1) edge[bend left=0] (n4);
		\draw[-latex] (n1) edge[bend left=5] (n5);

		\draw[-latex] (n2) edge[bend left=0] (n1);
		\draw[-latex] (n2) edge[bend left=0] (n3);
		\draw[-latex] (n2) edge[bend left=10] (n5);

		\draw[-latex] (n3) edge[bend left=0] (n6);

		\draw[-latex] (n4) edge[bend left=10] (n5);

		\draw[-latex] (n5) edge[bend left=5] (n1);
		\draw[-latex] (n5) edge[bend left=10] (n2);
		\draw[-latex] (n5) edge[bend left=10] (n4);

		\draw[-latex] (n1) to [in=240,out=300, loop]  node[auto] {} (n1);
		\draw[-latex] (n2) to [in=120, out=60, loop]  node[auto] {} (n2);						
		\draw[-latex] (n3) to [out=240,in=300, loop]  node[auto] {} (n3);			
		\draw[-latex] (n5) to [out=120, in=60, loop]  node[auto] {} (n5);	

		\draw[-latex] (n6) to [out=345, in=285, loop]  node[auto] {} (n6);

		\draw[-latex, color=red] (I1) edge[bend left=0] (n1);		
		\draw[-latex, color=red] (I2) edge[bend left=0] (n2);
		\draw[-latex, color=red] (n3) edge[bend left=0] (O1);
		\draw[-latex, color=red] (n4) edge[bend left=0] (O2);
		
	\end{tikzpicture}
\end{minipage}
\begin{minipage}[h]{50mm}
$\begin{pmatrix}[r]
 	1&0&1&1&1&0\\
	1&1&1&0&1&0\\
	0&0&1&0&0&1\\
	0&0&0&0&1&0\\
	1&1&0&1&1&0\\
	0&0&0&0&0&1
\end{pmatrix}$
\end{minipage}
\caption{$\mathfrak{G}_{\text{RL}}$: The new RIGHT-LEFT-XOR-gadget (RL-XOR-gadget). The red lines are not directly part of the gadget. As it can be seen, the gadget has circular planar embedding according to the nodes $1,2,3,4$.}
\label{RLXORgadget}
\end{figure}
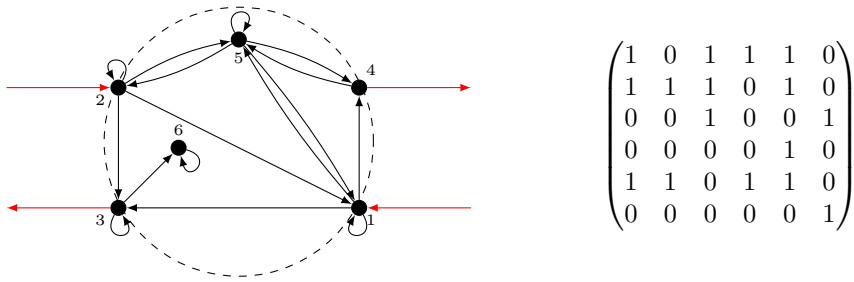
It has $1$, $2$ as input nodes and $3$, $4$ as its output nodes and all these nodes are boundary nodes. The \textit{correct} path is from $1$ to $3$ and $2$ to $4$. Its signature is
\begin{equation}\label{reqRLXORgadget}
 \mathcal{S}_{\gad_{\text{RL}}} = (3,3,2,4,3,3) = (0,0,2,1,0,0)_3
\end{equation}

\subsubsection{The LR-XOR-gadget} The second gadget we found can be seen in Figure \ref{LRXORgadget}. 

\begin{figure}[h!]
\centering
\begin{minipage}[h]{80mm}
	\centering
	\begin{tikzpicture}
		[scale=.8]
		\tikzstyle{every loop} = [-latex]
		\node[circle, draw, inner sep = 40pt, dashed] (circ) at (3,2.5) {};

		\node[circle,  inner sep = 2pt, thick] (O1) at (7,1) {};
		\node[circle,  inner sep = 2pt, thick] (O2) at (7,4) {};
		\node[circle,  inner sep = 2pt, thick] (I2) at (-1,4) {};
		\node[circle,  inner sep = 2pt, thick] (I1) at (-1,1) {};
	
		\node[circle, fill, draw, inner sep = 2pt] (n1) at (1,1){}; \node[] at (1,0.7) {\tiny{$1$}};
		\node[circle, fill, draw, inner sep = 2pt] (n2) at (1,4){}; \node[] at (1.2,3.8) {\tiny{$2$}};
		\node[circle, fill, draw, inner sep = 2pt] (n3) at (5,1){}; \node[] at (5.25,0.8) {\tiny{$3$}};  
		\node[circle, fill, draw, inner sep = 2pt] (n4) at (5,4){}; \node[] at (4.8,3.8) {\tiny{$4$}};
		\node[circle, fill, draw, inner sep = 2pt] (n5) at (3,2){}; \node[] at (3.2,2.2) {\tiny{$5$}};
		\node[circle, fill, draw, inner sep = 2pt] (n6) at (3,3.5){}; \node[] at (3,3.8) {\tiny{$6$}};
		
		\draw[-latex] (n1) edge[bend left=0] (n2);
		\draw[-latex] (n1) edge[bend left=7] (n3);				
		\draw[-latex] (n1) edge[bend left=0] (n5);		
		\draw[-latex] (n1) edge[bend left=7] (n6);		

		\draw[-latex] (n2) to [out=120, in=60, loop]  node[auto] {} (n2);						
		\draw[-latex] (n2) edge[bend left=0] (n4);

		\draw[-latex] (n3) edge[bend left=7] (n1);
		\draw[-latex] (n3) to [in=240,out=300, loop]  node[auto] {} (n3);							
		\draw[-latex] (n3) edge[bend left=0] (n5);
		\draw[-latex] (n3) edge[bend left=7] (n6);

		\draw[-latex] (n4) edge[bend left=0] (n3);		
		\draw[-latex] (n4) to [out=120, in=60, loop]  node[auto] {} (n4);
		\draw[-latex] (n4) edge[bend left=0] (n6);		

		\draw[-latex] (n5) to [in=240,out=300, loop]  node[auto] {} (n5);
		\draw[-latex] (n5) edge[bend left=7] (n6);	

		\draw[-latex] (n6) edge[bend left=7] (n1);	
		\draw[-latex] (n6) edge[bend left=7] (n3);
		\draw[-latex] (n6) edge[bend left=7] (n5);
	
		\draw[-latex, color=red] (I1) edge[bend left=0] (n1);		
		\draw[-latex, color=red] (I2) edge[bend left=0] (n2);
		\draw[-latex, color=red] (n3) edge[bend left=0] (O1);
		\draw[-latex, color=red] (n4) edge[bend left=0] (O2);
		
	\end{tikzpicture}
\end{minipage}
\begin{minipage}[h]{50mm}
	\centering
	$\begin{pmatrix}[r]
	0 & 1 & 1 & 0 & 1 & 1\\
	0 & 1 & 0 & 1 & 0 & 0\\
	1 & 0 & 1 & 0 & 1 & 1\\
	0 & 0 & 1 & 1 & 0 & 1\\
	0 & 0 & 0 & 0 & 1 & 1\\
	1 & 0 & 1 & 0 & 1 & 0
	\end{pmatrix}$
\end{minipage}
\caption{$\mathfrak{G}_{\text{LR}}$: The new LEFT-RIGHT-XOR-gadget (LR-XOR-gadget). The red lines are not directly part of the gadget. As it can be seen, the gadget has circular planar embedding according to the nodes $1,2,3,4$.}
\label{LRXORgadget}
\end{figure}
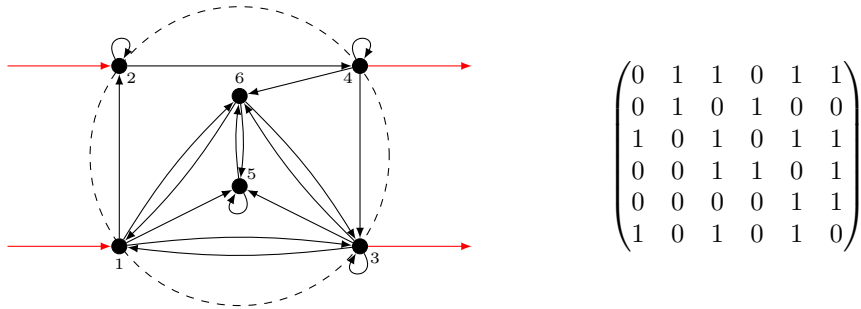
Equivalent to the RL-XOR-gadget, it has $1$,$2$ as input nodes and $3$,$4$ as its output nodes and all these nodes are boundary nodes. The \textit{correct} path is from $1$ to $3$ and $2$ to $4$. Its signature is
\begin{equation}\label{reqLRXORgadget}
\mathcal{S}_{\gad_{\text{LR}}} = (12,3,5,7,3,3) = (0,0,2,1,0,0)_3
\end{equation}

\begin{figure}
\centering
\begin{tikzpicture}
	[scale=.8]
	\node[rectangle, draw, minimum width = 130pt, minimum height = 40pt, rounded corners=2pt] (C1) at (0,5.5) {};
	\node[rectangle, minimum width = 50pt, minimum height = 20pt, rounded corners=2pt] (textC1) at (-0.5,6.7) {Clause 1: $x_1 \vee x_2 \vee x_3$};
	\node[rectangle, draw, fill=white, inner sep = 2pt] (CGO1) at (0,6) {$\mathfrak{G}_1$};
	\node[rectangle, draw, fill=white, inner sep = 2pt] (XGO101) at (-2,5) {$\mathfrak{G}_4$};
	\node[rectangle, draw, fill=white, inner sep = 2pt] (XGO102) at (0,5) {$\mathfrak{G}_4$};
	\node[rectangle, draw, fill=white, inner sep = 2pt] (XGO103) at (2,5) {$\mathfrak{G}_4$};

	\node[rectangle, draw, minimum width = 130pt, minimum height = 40pt, rounded corners=2pt] (C2) at (6,5.5) {};
	\node[rectangle, minimum width = 50pt, minimum height = 20pt, rounded corners=2pt] (textC2) at (5.5,6.7) {Clause 2: $x_1 \vee \neg x_4 \vee x_6$};
	\node[rectangle, draw, fill=white, inner sep = 2pt] (CGO2) at (6,6) {$\mathfrak{G}_1$};
	\node[rectangle, draw, fill=white, inner sep = 2pt] (XGO201) at (4,5) {$\mathfrak{G}_4$};
	\node[rectangle, draw, fill=white, inner sep = 2pt] (XGO202) at (6,5) {$\mathfrak{G}_4$};
	\node[rectangle, draw, fill=white, inner sep = 2pt] (XGO203) at (8,5) {$\mathfrak{G}_4$};

	\node[rectangle, draw, minimum width = 130pt, minimum height = 40pt, rounded corners=2pt] (C3) at (6,0.5) {};
	\node[rectangle, minimum width = 50pt, minimum height = 20pt, rounded corners=2pt] (textC3) at (5.5,-0.7) {Clause 3: $\neg x_1 \vee \neg x_3 \vee x_5$};
	\node[rectangle, draw, fill=white, inner sep = 2pt] (CGO3) at (6,0) {$\mathfrak{G}_1$};
	\node[rectangle, draw, fill=white, inner sep = 2pt] (XGO301) at (4,1) {$\mathfrak{G}_3$};
	\node[rectangle, draw, fill=white, inner sep = 2pt] (XGO302) at (6,1) {$\mathfrak{G}_3$};
	\node[rectangle, draw, fill=white, inner sep = 2pt] (XGO303) at (8,1) {$\mathfrak{G}_3$};

	\node[rectangle, draw, minimum width = 130pt, minimum height = 40pt, rounded corners=2pt] (C4) at (12,5.5) {};
	\node[rectangle, minimum width = 50pt, minimum height = 20pt, rounded corners=2pt] (textC4) at (11.5,6.7) {Clause 4: $\neg x_4 \vee \neg x_5 \vee x_6$};
	\node[rectangle, draw, fill=white, inner sep = 2pt] (CGO4) at (12,6) {$\mathfrak{G}_1$};
	\node[rectangle, draw, fill=white, inner sep = 2pt] (XGO401) at (10,5) {$\mathfrak{G}_4$};
	\node[rectangle, draw, fill=white, inner sep = 2pt] (XGO402) at (12,5) {$\mathfrak{G}_4$};
	\node[rectangle, draw, fill=white, inner sep = 2pt] (XGO403) at (14,5) {$\mathfrak{G}_4$};

	\node[rectangle, draw, fill=white, inner sep = 2pt] (VG1) at (-2,3) {$\;\mathfrak{G}_2:\;x_1$};
	\node[rectangle, draw, fill=white, inner sep = 2pt] (VG2) at (1,3) {$\;\mathfrak{G}_2:\;x_2$};
	\node[rectangle, draw, fill=white, inner sep = 2pt] (VG3) at (4,3) {$\;\mathfrak{G}_2:\;x_3$};
	\node[rectangle, draw, fill=white, inner sep = 2pt] (VG4) at (7,3) {$\;\mathfrak{G}_2:\;x_4$};
	\node[rectangle, draw, fill=white, inner sep = 2pt] (VG5) at (10,3) {$\;\mathfrak{G}_2:\;x_5$};
	\node[rectangle, draw, fill=white, inner sep = 2pt] (VG6) at (13,3) {$\;\mathfrak{G}_2:\;x_6$};
	
	\draw[-latex, color=red] (CGO1) to [out=170, in=170]  node[auto] {} (XGO101);
	\draw[-latex, color=red] (XGO101) to [out=10, in=190]  node[auto] {} (CGO1);
	\draw[-latex, color=red] (CGO1) to [out=192, in=170]  node[auto] {} (XGO102);
	\draw[-latex, color=red] (XGO102) to [out=10, in=340]  node[auto] {} (CGO1);
	\draw[-latex, color=red] (CGO1) to [out=342, in=170]  node[auto] {} (XGO103);
	\draw[-latex, color=red] (XGO103) to [out=10, in=2]  node[auto] {} (CGO1);

	\draw[-latex, color=red] (CGO2) to [out=170, in=170]  node[auto] {} (XGO201);
	\draw[-latex, color=red] (XGO201) to [out=10, in=190]  node[auto] {} (CGO2);
	\draw[-latex, color=red] (CGO2) to [out=192, in=170]  node[auto] {} (XGO202);
	\draw[-latex, color=red] (XGO202) to [out=10, in=340]  node[auto] {} (CGO2);
	\draw[-latex, color=red] (CGO2) to [out=342, in=170]  node[auto] {} (XGO203);
	\draw[-latex, color=red] (XGO203) to [out=10, in=2]  node[auto] {} (CGO2);

	\draw[-latex, color=red] (CGO3) to [out=342, in=342]  node[auto] {} (XGO303);
	\draw[-latex, color=red] (XGO303) to [out=192, in=2]  node[auto] {} (CGO3);
	\draw[-latex, color=red] (CGO3) to [out=4, in=342]  node[auto] {} (XGO302);
	\draw[-latex, color=red] (XGO302) to [out=192, in=170]  node[auto] {} (CGO3);
	\draw[-latex, color=red] (CGO3) to [out=172, in=342]  node[auto] {} (XGO301);
	\draw[-latex, color=red] (XGO301) to [out=192, in=190]  node[auto] {} (CGO3);

	\draw[-latex, color=red] (CGO4) to [out=170, in=170]  node[auto] {} (XGO401);
	\draw[-latex, color=red] (XGO401) to [out=10, in=190]  node[auto] {} (CGO4);
	\draw[-latex, color=red] (CGO4) to [out=192, in=170]  node[auto] {} (XGO402);
	\draw[-latex, color=red] (XGO402) to [out=10, in=340]  node[auto] {} (CGO4);
	\draw[-latex, color=red] (CGO4) to [out=342, in=170]  node[auto] {} (XGO403);
	\draw[-latex, color=red] (XGO403) to [out=10, in=2]  node[auto] {} (CGO4);

	\draw[-latex, color=red] (VG1) to [out=70, in=190]  node[auto] {} (XGO101);
	\draw[-latex, color=red] (XGO101) to [out=342, in=30]  node[auto] {} (VG1);
	\draw[-latex, color=red, dashed] (VG1) to [out=120, in=190]  node[auto] {} (XGO201);
	\draw[-latex, color=red, dashed] (XGO201) to [out=342, in=80]  node[auto] {} (VG1);
	\draw[-latex, color=red] (VG1) to [out=250, in=170]  node[auto] {} (XGO301);
	\draw[-latex, color=red] (XGO301) to [out=13, in=310]  node[auto] {} (VG1);	

	\draw[-latex, color=red] (VG2) to [out=120, in=190]  node[auto] {} (XGO102);
	\draw[-latex, color=red] (XGO102) to [out=342, in=80]  node[auto] {} (VG2);

	\draw[-latex, color=red] (VG3) to [out=120, in=190]  node[auto] {} (XGO103);
	\draw[-latex, color=red] (XGO103) to [out=342, in=80]  node[auto] {} (VG3);	
	\draw[-latex, color=red] (VG3) to [out=230, in=170]  node[auto] {} (XGO302);
	\draw[-latex, color=red] (XGO302) to [out=10, in=270]  node[auto] {} (VG3);

	\draw[-latex, color=red] (VG4) to [out=150, in=190]  node[auto] {} (XGO202);
	\draw[-latex, color=red] (XGO202) to [out=342, in=130]  node[auto] {} (VG4);
	\draw[-latex, color=red] (VG4) to [out=90, in=190]  node[auto] {} (XGO401);
	\draw[-latex, color=red] (XGO401) to [out=342, in=60]  node[auto] {} (VG4);

	\draw[-latex, color=red] (VG5) to [out=120, in=190]  node[auto] {} (XGO402);
	\draw[-latex, color=red] (XGO402) to [out=342, in=80]  node[auto] {} (VG5);
	\draw[-latex, color=red] (VG5) to [out=210, in=170]  node[auto] {} (XGO303);
	\draw[-latex, color=red] (XGO303) to [out=10, in=230]  node[auto] {} (VG5);

	\draw[-latex, color=red, dashed] (VG6) to [out=80, in=190]  node[auto] {} (XGO203);
	\draw[-latex, color=red, dashed] (XGO203) to [out=342, in=50]  node[auto] {} (VG6);
	\draw[-latex, color=red] (VG6) to [out=150, in=190]  node[auto] {} (XGO403);
	\draw[-latex, color=red] (XGO403) to [out=342, in=130]  node[auto] {} (VG6);
	
\end{tikzpicture} 
\caption{This is a pn-planar formula $\Phi = (x_1 \vee x_2 \vee x_3) \wedge (x_1 \vee \neg x_4 \vee x_6) \wedge  (\neg x_1 \vee \neg x_3 \vee x_5) \wedge (\neg x_4 \vee \neg x_5 \vee x_6)$ with its associated Valiant graph. The clauses 1,2 and 4 are on face 1 and the clause 3 is on face 2. The dashed connection are the two connection from Figure \ref{twoDrawings}, that actually circumvent the clause nodes 1 and 4.}
\label{valiantGraph}
\end{figure}
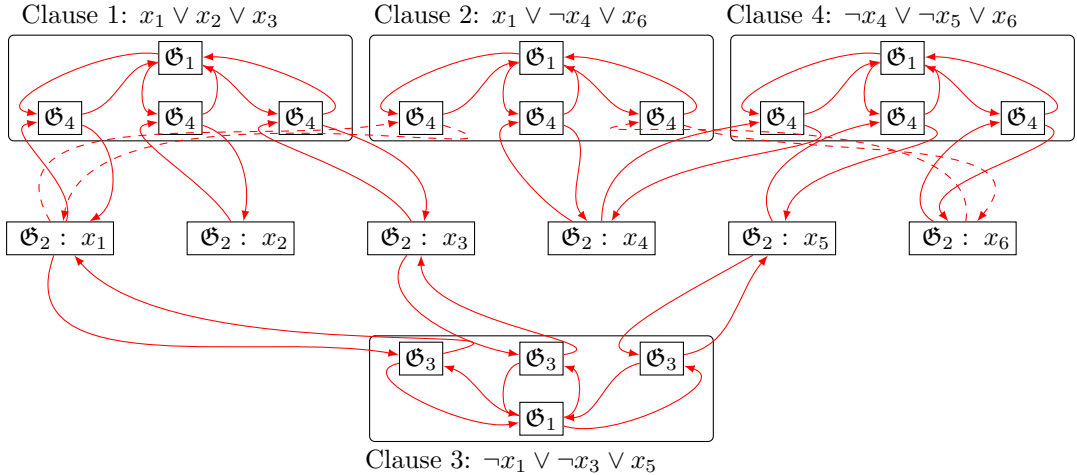

The graph shown in Figure \ref{valiantGraph} is the Valiant graph of the pn-planar formula from Figure \ref{twoDrawings}. It can be seen, that the graph is indeed planar. Note that on face 1 the gadget $\mathfrak{G}_3$ is used as the XOR-gadget and on face 2 it is $\mathfrak{G}_4$. This is important, since otherwise the input and output edges to the VARIABLE-gadget would cross. 

\begin{lemma}\label{lemma1}
 Using the gadgets $\mathfrak{G}_1,\mathfrak{G}_2,\mathfrak{G}_{\text{RL}}$ and $\mathfrak{G}_{\text{LR}}$ and the reduction from $\Phi$ to $\Phi'$, whereof $\Phi'$ is pn-planar, one could could make the graph of Theorem \ref{theorem2a}/\ref{theorem2b} planar.  
\end{lemma}

\begin{proof}[proof of Lemma \ref{lemma1}]
 To proof this, one can take a look at the example drawing of the pn-planar formula Valiant graph in Figure \ref{valiantGraph}. It is not hard to see, how this graph could be redrawn using the gadgets $\mathfrak{G}_1,\mathfrak{G}_2,\mathfrak{G}_{\text{RL}},\mathfrak{G}_{\text{LR}}$. 
\end{proof}


\section{Circular planar bipartite double covers and a limiting conjecture}  We now turn towards generalization. The bipartite double cover operation turns out to play a central role and we are mainly interested in the way it changes the planarity of a graph. For our approach, a mandatory requirement is, that the used gadgets itself must turn into planar bipartite double covers in order to have chance to achieve planarity for the whole graph. And simple planarity alone is not enough. Similar to the RL and LR-XOR-gadgets, that achieve circular planarity for their input/output nodes, one now has to take care about the circular planarity of their bipartite double cover regarding the input/output nodes.

\begin{lemma}[Bipartite connectable graphs]\label{lConnectable}
 If a planar graph $\grG$ can be drawn as a connection of gadgets $\gad_i$, then $\hat{\grG}$ will be also planar if each gadget itself has a circular planar bipartite double cover $\hat{\gad}_i$, with the input and output nodes on the boundary in the same order as in $\gad$ itself.
\end{lemma}

\begin{proof} The lemma is immediately clear if one considers the example gadget in Figure \ref{example}. The graph on the left side is the original graph $\grG$ and on the right side its bipartite double cover. If one builds a graph $\grG$ using gadgets $\gad_i$ (like shown on the left side in the figure), one simply replaces each of occurrences of $\gad_i$ with $\hat{\gad}_i$ (the one on the right side in the figure). Since the input and output pairs of $\hat{\gad}_i$ are in the same order and also on the boundary, the original planar structure between the gadgets is overtaken and hence $\hat{\grG}$ will be planar if $\grG$ is.
\begin{figure}
	\centering 
	\begin{minipage}[h]{50mm}
		\centering 
		\begin{tikzpicture}
			[scale=.8]
			\node[circle, draw, inner sep = 34pt, dashed] (circ) at (1.5,1.5) {};
			\node[circle, fill, draw, inner sep = 2pt] (n1) at (0,3){}; \node[] at (0,3.3) {\tiny{$1$}};
				\node[rectangle, rotate=0, inner sep = 2pt] (h1) at (1.5,3.8) {...};
			\node[circle, fill, draw, inner sep = 2pt] (n2) at (1.5,3){}; \node[] at (1.5,3.3) {\tiny{$2$}};
			\node[circle, fill, draw, inner sep = 2pt] (n3) at (3,3){}; \node[] at (3,3.3) {\tiny{$3$}};
		
			\node[circle, fill, draw, inner sep = 2pt] (n4) at (0,0){}; \node[] at (0,-0.3) {\tiny{$4$}};
			\node[circle, fill, draw, inner sep = 2pt] (n5) at (1.5,0){}; \node[] at (1.5,-0.3) {\tiny{$5$}};
				\node[rectangle, rotate=0, inner sep = 2pt] (h2) at (1.5,-0.8) {...};
			\node[circle, fill, draw, inner sep = 2pt] (n6) at (3,0){}; \node[] at (3,-0.3) {\tiny{$6$}};

			\draw[-,color=red] (n1) edge[bend left=30] (h1); \draw[-latex,color=red] (h1) edge[bend left=30] (n3);
			\draw[-,color=red] (n6) edge[bend left=30] (h2); \draw[-latex,color=red] (h2) edge[bend left=30] (n4); 
		
			\draw[-latex] (n3) edge[bend right=0] (n6);  
			\draw[-latex] (n4) edge[bend right=0] (n1);
			
			\draw[-latex] (n1) edge[bend right=15] (n2);  
			\draw[-latex] (n2) edge[bend right=15] (n1);
		
			\draw[-latex] (n2) edge[bend right=15] (n3);  
			\draw[-latex] (n3) edge[bend right=15] (n2); 
		
			\draw[-latex] (n3) edge[bend left=40] (n1); 	
	
			\draw[-latex] (n4) edge[bend right=15] (n5);  
			\draw[-latex] (n5) edge[bend right=15] (n4);
		
			\draw[-latex] (n5) edge[bend right=15] (n6);  
			\draw[-latex] (n6) edge[bend right=15] (n5); 
		
			\draw[-latex] (n6) edge[bend right=40] (n4); 			
		\end{tikzpicture}
	\end{minipage}
	\begin{minipage}[h]{50mm}
		\centering 
		\begin{tikzpicture}
			[scale=.8]
			\node[circle, draw, inner sep = 26pt, dashed] (circ) at (1,1) {};
			\node[circle, inner sep = 2pt] (dots91) at (1,3) {\tiny{...}};
			\node[circle, inner sep = 2pt] (dots106) at (1,-1) {\tiny{...}};
	
			\node[circle, fill, draw, inner sep = 2pt] (n1) at (0,2.25){}; \node[] at (-0.3,2.25) {\tiny{$1$}};
			\node[circle, fill, draw, inner sep = 2pt] (n2) at (2,1.75){}; \node[] at (2.3,1.75) {\tiny{$2$}};
			\node[circle, fill, draw, inner sep = 2pt] (n3) at (0,1.25){}; \node[] at (-0.3,1.25) {\tiny{$3$}};
			\node[circle, fill, draw, inner sep = 2pt] (n4) at (2,0.75){}; \node[] at (2.3,0.75) {\tiny{$4$}};
			\node[circle, fill, draw, inner sep = 2pt] (n5) at (0,0.25){}; \node[] at (-0.3,0.25) {\tiny{$5$}};
			\node[circle, fill, draw, inner sep = 2pt] (n6) at (2,-0.25){}; \node[] at (2.3,-0.25) {\tiny{$6$}};
	
			\node[circle, fill, draw, inner sep = 2pt] (n7) at (2,1.25){}; \node[] at (2.3,1.25) {\tiny{$7$}};
			\node[circle, fill, draw, inner sep = 2pt] (n8) at (0,1.75){}; \node[] at (-0.3,1.75) {\tiny{$8$}};
			\node[circle, fill, draw, inner sep = 2pt] (n9) at (2,2.25){}; \node[] at (2.3,2.25) {\tiny{$9$}};
			\node[circle, fill, draw, inner sep = 2pt] (n10) at (0,-0.25){}; \node[] at (-0.3,-0.25) {\tiny{$10$}};
			\node[circle, fill, draw, inner sep = 2pt] (n11) at (2,0.25){}; \node[] at (2.3,0.25) {\tiny{$11$}};
			\node[circle, fill, draw, inner sep = 2pt] (n12) at (0,0.75){}; \node[] at (-0.3,0.75) {\tiny{$12$}};
	
			\draw[-] (n1) to (n8);
	
			\draw[-] (n2) to (n7);
			\draw[-] (n2) to (n9);
	
			\draw[-] (n3) to (n7); 
			\draw[-] (n3) to (n8); 
			\draw[-] (n3) to (n12); 
			
			\draw[-] (n4) to (n7);
			\draw[-] (n4) to (n11);
			
			\draw[-] (n5) to (n10);
			\draw[-] (n5) to (n12);
	
			\draw[-] (n6) to (n10); 
			\draw[-] (n6) to (n11); 
	
			\draw[-, color = red] (n1) edge[bend left=35] (dots91);\draw[-latex, color = red] (dots91) edge[bend left=35] (n9);
			\draw[-, color = red] (n6) edge[bend left=35] (dots106);\draw[-latex, color = red] (dots106) edge[bend left=35] (n10);
			
		\end{tikzpicture}
	\end{minipage}
\caption{On the left the original example graph, with input and output pairs $(3,1)$ and $(4,6)$. On the right the bipartite double cover with corresponding input and output pair $(9,1)$ and $(10,6)$. It can be connected in the same way as the original graph.}
\label{example}
\end{figure}
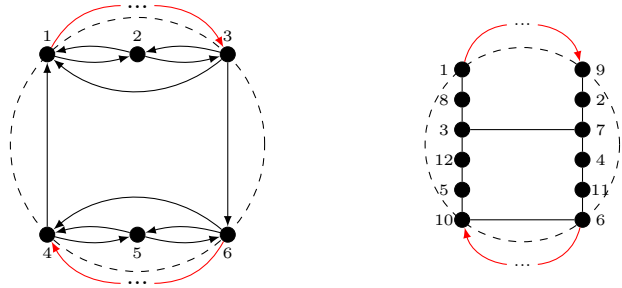
\end{proof}
Lemma \ref{lConnectable} states that one could draw the bipartite double cover of a graph by simply replacing each gadget by its BDC, if each gadget turns into suitable bipartite double cover as well. Knowing this and remembering that we already use planar formulas, in particular pn-planar formulas, as the input, actually we only need to find the correct gadgets to make the BDC planar. We do not have to care about the structure of the input formulas, but only that the gadgets itself have a appropriate bipartite double cover. And since gadgets could be rather small compared to the overall graph, this seems to be a more accessible approach. However, we will show in the remaining sections, that it is not that easy at all.

First, we take a look at the bipartite double cover of the CLAUSE-gadget, cf. Figure \ref{clauseGadgetBDC}.
\begin{figure}[h!]
\centering
\begin{minipage}[h]{50mm}
	\centering
	\begin{tikzpicture}
		[scale=.8]
		\node[circle, fill, draw, inner sep = 2pt] (n1) at (0,1){}; \node[] at (-0.2,0.8) {\tiny{$1$}};
			\node[rectangle, inner sep = 2pt] (xor12) at (1,0) {...};
		\node[circle, fill, draw, inner sep = 2pt] (n2) at (2,1){}; \node[] at (2.2,0.8) {\tiny{$2$}};
			\node[rectangle, rotate=-68, inner sep = 2pt] (xor23) at (2.5,2.4) {...};
		\node[circle, fill, draw, inner sep = 2pt] (n3) at (1,3){}; \node[] at (1,3.3) {\tiny{$3$}};
			\node[rectangle, rotate=68,  inner sep = 2pt] (xor31) at (-0.5,2.4) {...};
		\node[circle, fill, draw, inner sep = 2pt] (n4) at (1,1.8){}; \node[] at (1.2,2) {\tiny{$4$}};	
	
		\draw[-,color=red] (n3) edge[bend right=40] (xor31); \draw[-latex,color=red] (xor31) edge[bend right=40] (n1);
		\draw[-,color=red] (n2) edge[bend right=40] (xor23); \draw[-latex,color=red] (xor23) edge[bend right=40] (n3); 
		\draw[-,color=red] (n1) edge[bend right=40] (xor12); \draw[-latex,color=red] (xor12) edge[bend right=40] (n2);
	
		\draw[-latex, dashed] (n4) edge[bend right=15] (n3);  
		\draw[-latex, dashed] (n3) edge[bend right=15] (n4);
	
		\draw[-latex, dashed] (n4) edge[bend right=15] (n1);  
		\draw[-latex, dashed] (n1) edge[bend right=15] (n4); 
	
		\draw[-latex, dashed] (n2) edge[bend right=15] (n4); 
		\draw[-latex, dashed] (n4) edge[bend right=15] (n2);
	
		\draw[-latex, dashed] (n1) edge[bend right=15] (n2);
		\draw[-latex, dashed] (n2) edge[bend right=15] (n1);
	\end{tikzpicture}
\end{minipage}
\begin{minipage}[h]{50mm}
	\centering
	\begin{tikzpicture}
		[scale=.8]
		\node[circle, draw, inner sep = 25pt, dashed] (circ) at (0.5,2) {};

		\node[circle, inner sep = 2pt] (dotsA2) at (0.5,-0.5) {\tiny{...}};
		\node[circle, inner sep = 2pt, rotate=45] (dotsB3) at (-0.8,4) {\tiny{...}};
		\node[circle, inner sep = 2pt, rotate=-45] (dotsC1) at (1.8,4) {\tiny{...}};

		\node[circle, fill, draw, inner sep = 2pt] (n2) at (0,0.5){}; \node[] at (-0.2,0.3) {\tiny{$6$}};
		\node[circle, fill, draw, inner sep = 2pt] (nD) at (0,2){}; \node[] at (0.2,2.2) {\tiny{$4$}};
		\node[circle, fill, draw, inner sep = 2pt] (n3) at (-1,2){}; \node[] at (-1.3,2) {\tiny{$7$}};
		\node[circle, fill, draw, inner sep = 2pt] (n1) at (0,3.5){}; \node[] at (0,3.8) {\tiny{$5$}};
		\node[circle, fill, draw, inner sep = 2pt] (n4) at (1,2){}; \node[] at (1.2,2.2) {\tiny{$8$}};
		\node[circle, fill, draw, inner sep = 2pt] (nB) at (1,3.5){}; \node[] at (1,3.8) {\tiny{$2$}};
		\node[circle, fill, draw, inner sep = 2pt] (nC) at (2,2){}; \node[] at (2.3,2) {\tiny{$3$}};
		\node[circle, fill, draw, inner sep = 2pt] (nA) at (1,0.5){}; \node[] at (1.2,0.3) {\tiny{$1$}};
	
		\draw[-, color = red] (nA) edge[bend left=35] (dotsA2);\draw[-latex, color = red] (dotsA2) edge[bend left=35] (n2);
		\draw[-, color = red] (nB) edge[bend right=35] (dotsB3);\draw[-latex, color = red] (dotsB3) edge[bend right=35] (n3);
		\draw[-, color = red] (nC) edge[bend right=25] (dotsC1);\draw[-latex, color = red] (dotsC1) edge[bend right=25] (n1);

		\draw[-] (nA) to (n2);
		\draw[-] (nA) to (n4);
		\draw[-] (nB) to (n1);
		\draw[-] (nB) to (n4);
		\draw[-] (nC) to (n4);
		\draw[-] (nD) to (n1);
		\draw[-] (nD) to (n2);
		\draw[-] (nD) to (n3);

	\end{tikzpicture}
\end{minipage}
\caption{On the left the CLAUSE-gadget and on the right the bipartite double cover with the boundary nodes on the dashed disc. The connections (7,2) and (5,3) intersect, which destroys the overall planarity of the graph and makes the original CLAUSE-gadget not usable.}
\label{clauseGadgetBDC}
\end{figure}
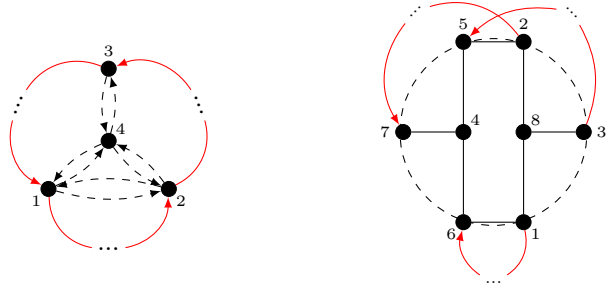

As it can be seen, the CLAUSE-gadget has a circular planar embedding for its bipartite double cover, but with a wrong order. The order must contain the nodes of the pairs (5,3), (6,1) and (7,2) as successive nodes. Every interleaving destroys the planarity of the graph. So Valiant's original CLAUSE-gadget does not have the correct order regarding its boundary nodes when building its bipartite double cover. 

\textbf{Remark:} If we talk about boundary nodes and their order, we always only consider the input and output nodes on the boundary like drawn in Figure \ref{clauseGadgetBDC}. One could easily also increase the boundary by the nodes $4$ and $8$, but as they are not used as input and output nodes, we keep them as interior nodes. 


\begin{defin}
 Let $\gad$ be a gadget with input/output nodes $(i_1,o_1)$ and $(i_2,o_2)$. 
 If the bipartite double cover is circular planar according to the nodes $i'_1,i'_2,o_1,o_2$, then it is called
 \begin{itemize}
  \item \texttt{connectable} if the order is $i'_1,o_1,i'_2,o_2$
  \item \texttt{extendable} if the order is $i'_1,i'_2,o_2,o_1$ $\blacktriangleleft$
 \end{itemize}
\end{defin}
Note that an order of $i'_1,i'_2,o_1,o_2$ can not exists if the gadget is planar. Regarding the two new XOR-gadgets it holds that, $\mathfrak{G}_{\text{RL}}$ is connectable and $\mathfrak{G}_{\text{LR}}$ is extendable. They both have a planar bipartite cover, but are not circular planar according to the input/output nodes. The only gadget that turns out to be usable, i.e., that also is circular planar with the correct boundary order, is the VARIABLE-gadget. It has the following bipartite double cover (see Figure \ref{bdcVARIABLEgadget})

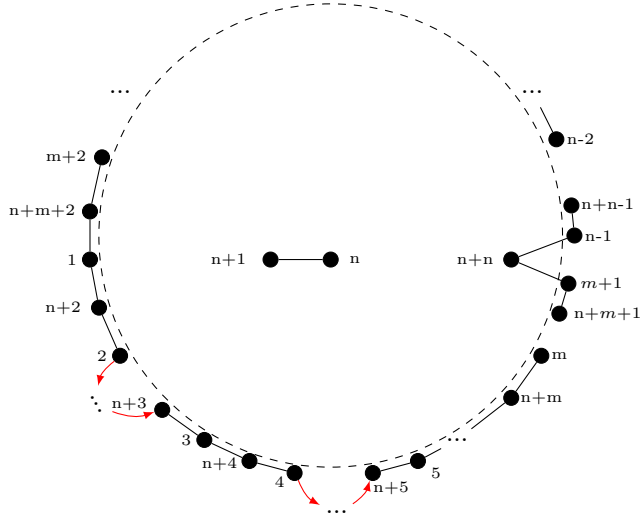
\begin{figure}
	\begin{tikzpicture}
		[scale=.8]
		\node[circle, draw, inner sep = 62pt, dashed] (circ) at (4,2.4) {};

		\node[circle, fill, draw, inner sep = 2pt] (Nn1) at (3,2){}; \node[] at (2.3,2) {\tiny{n+$1$}};
		\node[circle, fill, draw, inner sep = 2pt] (nn) at (4,2){}; \node[] at (4.4,2) {\tiny{n}};
		
		\node[circle, inner sep = 2pt] (dots2) at (7.35,4.8) {...};
		\node[circle, fill, draw, inner sep = 2pt] (nnM2) at (7.75,4){}; \node[] at (8.15,4) {\tiny{n-$2$}};
		\node[circle, fill, draw, inner sep = 2pt] (NnnM1) at (8,2.9){}; \node[] at (8.6,2.9) {\tiny{n+n-$1$}};
		\node[circle, fill, draw, inner sep = 2pt] (nnM1) at (8.05,2.4){}; \node[] at (8.45,2.4) {\tiny{n-1}};

		\node[circle, inner sep = 2pt] (dots3) at (0.5,4.8) {...};
		\node[circle, fill, draw, inner sep = 2pt] (nmP2) at (0.2,3.7){}; \node[] at (-0.4,3.7) {\tiny{m+$2$}};
		\node[circle, fill, draw, inner sep = 2pt] (NnmP2) at (0,2.8){}; \node[] at (-0.8,2.8) {\tiny{n+m+$2$}};

		\node[circle, fill, draw, inner sep = 2pt] (n1) at (0,2){}; \node[] at (-0.3,2) {\tiny{1}};

		\node[circle, fill, draw, inner sep = 2pt] (Nn2) at (0.15,1.2){}; \node[] at (-0.45,1.2) {\tiny{n+$2$}};
		\node[circle, fill, draw, inner sep = 2pt] (n2) at (0.5,0.4){}; \node[] at (0.2,0.4) {\tiny{2}};
			\node[circle, inner sep = 2pt, rotate=120] (dotsC1) at (0.1,-0.4) {...};
		\node[circle, fill, draw, inner sep = 2pt] (Nn3) at (1.2,-0.5){}; \node[] at (0.65,-0.4) {\tiny{n+$3$}};
		\node[circle, fill, draw, inner sep = 2pt] (n3) at (1.9,-1){}; \node[] at (1.6,-1) {\tiny{$3$}};
		\node[circle, fill, draw, inner sep = 2pt] (Nn4) at (2.65,-1.35){}; \node[] at (2.15,-1.4) {\tiny{n+$4$}};
		\node[circle, fill, draw, inner sep = 2pt] (n4) at (3.4,-1.55){}; \node[] at (3.15,-1.7) {\tiny{$4$}};
			\node[circle, inner sep = 2pt, rotate=0] (dotsC2) at (4.1,-2.2) {...};
		\node[circle, fill, draw, inner sep = 2pt] (Nn5) at (4.7,-1.55){}; \node[] at (5,-1.8) {\tiny{n+$5$}};
		\node[circle, fill, draw, inner sep = 2pt] (n5) at (5.45,-1.35){}; \node[] at (5.75,-1.55) {\tiny{$5$}};

		\node[circle, inner sep = 2pt] (dots) at (6.1,-1) {...};
		\node[circle, fill, draw, inner sep = 2pt] (Nnm) at (7,-0.3){}; \node[] at (7.5,-0.3) {\tiny{n+m}};
		\node[circle, fill, draw, inner sep = 2pt] (nm) at (7.5,0.4){}; \node[] at (7.8,0.4) {\tiny{m}};

		\node[circle, fill, draw, inner sep = 2pt] (NnmP1) at (7.8,1.1){}; \node[] at (8.6,1.1) {\tiny{n+$m$+$1$}};
		\node[circle, fill, draw, inner sep = 2pt] (nmP1) at (7.95,1.6){}; \node[] at (8.5,1.6) {\tiny{$m$+$1$}};
		
		\node[circle, fill, draw, inner sep = 2pt] (NnnPn) at (7,2){}; \node[] at (6.4,2) {\tiny{n+n}};
		
		\draw[-latex, color=red] (n2) edge[bend right=20] (dotsC1);\draw[-latex, color=red] (dotsC1) edge[bend right=20] (Nn3);
		\draw[-latex, color=red] (n4) edge[bend right=20] (dotsC2);\draw[-latex, color=red] (dotsC2) edge[bend right=20] (Nn5);

		\draw[-] (nnM2) to (dots2);
		\draw[-] (nnM1) to (NnnM1);
		\draw[-] (nnM1) to (NnnPn);
		\draw[-] (nn) to (Nn1);
		\draw[-] (NnmP2) to (nmP2);
		\draw[-] (n1) to (NnmP2);
		\draw[-] (n1) to (Nn2);
		\draw[-] (Nn2) to (n2);

		\draw[-] (Nn3) to (n3);
		\draw[-] (n3) to (Nn4);
		\draw[-] (Nn4) to (n4);

		\draw[-] (Nn5) to (n5);
		\draw[-] (n5) to (dots);
		\draw[-] (dots) to (Nnm);
		\draw[-] (nm) to (Nnm);

		\draw[-] (NnmP1) to (nmP1);
		\draw[-] (nmP1) to (NnnPn);

	\end{tikzpicture}
	\caption{The input/output pairs are $(n+3,2)$, $(n+5,4)$, ...,$(n+m+1,m)$ and $(n+m+3,m+2)$, $(n+m+5,m+4)$, ..., $(n+n-1,n-2)$. And they can be drawn at pairs on a circle.}
	\label{bdcVARIABLEgadget}
\end{figure}

The input/output pairs of this bipartite double cover are clockwise for the TRUE-nodes and anticlockwise for the FALSE-nodes. It can be swapped by swapping the circle horizontally.

The next two Lemmas are well known results regarding bipartite graphs.

\begin{lemma}\label{lBIP}
 If $\grG$ is a bipartite graph, then $\hat{\grG}$ is a disconnected graph that consists of two copies of $\grG$.
\end{lemma}

\begin{lemma}\label{ldis}
 A graph $\grG$ is bipartite if and only if $\hat{\grG}$ is disconnected.
\end{lemma}

There is a relationship between the signature of $\grG$ and the signature of $\hat{\grG}$.

\begin{lemma}\label{ryser}
 Let $\grG$ be a graph with adjacency matrix $\mA$ that has the signature
\begin{equation}
	\mathcal{S}_{\grG} = (r_1,r_2,r_3,r_4,r_5,r_6)
\end{equation}
Let $\mB$ be the matrix of the bipartite double cover $\hat{\grG}$, with $i'_1 = i_1 + \text{dim}(\mA)$ and $i'_2 = i_2 + \text{dim}(\mA)$. Then
\begin{equation}
	\mathcal{S}_{\hat{\grG}} = (r_1^2,r_1r_2,r_1r_3,r_1r_4,r_1r_5,r_1r_6)
\end{equation}
\end{lemma}
\begin{proof}
 The proof can be deduced from Ryser permanent formula for a $n\times n$ matrix $\mB$
 \begin{equation}
  \perm(\mB) = (-1)^n \sum_{S \subseteq \{1,...,n\}} (-1)^{|S|}\prod^n_{i=1}\sum_{j \in S}a_{ij}
 \end{equation}
 The dimension of our matrix is even since it is a bipartite matrix, i.e, we write $|\mB| = 2n$ so
 $\mB$ is of the form
 \begin{equation}
  \mB = \begin{pmatrix}[c] 
			& 		 &	&	&		&	\\
			& \bold{0^{n,n}} & 	&	& \mA		&	\\
			&		 &	&	&		&	\\ 
			& \mA^\tp 	 & 	&	&\bold{0^{n,n}} & 	\\
			& \;		 &	&	&		&	
	\end{pmatrix} 
 \end{equation} 
 with $\mA$ being a $n\times n$ matrix. Further note that $\perm(\mA) = \perm(\mA^\tp)$. With the rows/columns at
 \begin{equation}
    \mB = 
	\overset{\;\;\;\;\;\;\;\;\;\;\;\;\;\;\;\;\;\;\;\;i'_1\downarrow\;\;\;i'_2\downarrow\;\;}{ 
	\begin{pmatrix}[c] 
			& 		 &	&	&		&	\\
			& \bold{0^{n,n}} & 	&	& \mA		&	\\
			&		 &	&	&		&	\\ 
			& \mA^\tp 	 & 	&	&\bold{0^{n,n}} & 	\\
			& \;		 &	&	&		&	
	\end{pmatrix} 
  	}
	\begin{matrix}[c]
		\leftarrow o_1\\
		\leftarrow o_2\\
		\;\\
		\;
	\end{matrix}
 \end{equation}
 \begin{align*}
  \perm(\mB) & = (-1)^{2n-\tau} \sum_{S \subseteq \{1,...,2n-\tau \}} (-1)^{|S|}\prod^{2n-\tau}_{i=1}\sum_{j \in S}a_{ij}\\
	     & = (-1)^{2n-\tau} \sum_{S \subseteq \{1,...,2n-\tau \}} (-1)^{|S|}\left(\prod^{n-\tau/2}_{i=1}\sum_{j \in S}a_{ij} \cdot \prod^{2n-\tau/2}_{i=n+1-\tau/2}\sum_{j \in S}a_{ij} \right)
 \end{align*}
 with $\tau \in \{0,1,2\}$ according to the cases if zero, one pair or both pairs of $(i'_1,o_1),(i'_2,o_2)$ have been removed.  One can write each subset $S \subseteq \{1,...,2n\}$ as two unique subsets $S_1$, $S_2$, with $S_1 \cup S_2 = S$ and $S_1 \subseteq \{1,...,n-\tau\}$ and $S_2 \subseteq \{n+1-\tau,...,2n-\tau\}$, so $S_1 \cap S_2 = \varnothing$. Since $a_{ij} = 0$ for $
1 \leq i,j \leq n$ and $n+1 \leq i,j \leq 2n$ we can write
\begin{align*}
  \perm(\mB) & = (-1)^{2n-\tau} \sum_{S \subseteq \{1,...,2n-\tau \}} (-1)^{|S|}\left(\prod^{n-\tau}_{i=1}\sum_{j \in S_2}a_{ij} \cdot \prod^{2n-\tau}_{i=n+1-\tau}\sum_{j \in S_1}a_{ij} \right) \\
	     & = (-1)^{2n-\tau} \sum_{ S_1 \subseteq \{1,...,n-\tau \}, S_2 \subseteq \{n+1-\tau,...,2n-\tau \}  } (-1)^{|S_1 \cup S_2|}\left(\prod^{n-\tau}_{i=1}\sum_{j \in S_2}a_{ij} \cdot \prod^{2n-\tau}_{i=n+1-\tau}\sum_{j \in S_1}a_{ij} \right) \\
	     & = (-1)^{2n-\tau} \sum_{ S_1 \subseteq \{1,...,n-\tau \}, S_2 \subseteq \{n+1-\tau,...,2n-\tau \}  } (-1)^{|S_2|}\prod^{n-\tau}_{i=1}\sum_{j \in S_2}a_{ij} \cdot (-1)^{|S_1|} \prod^{2n-\tau}_{i=n+1-\tau}\sum_{j \in S_1}a_{ij}\\
	     & = (-1)^{2n-\tau} \sum_{ S_2 \subseteq \{n+1-\tau,...,2n-\tau \}  } (-1)^{|S_2|}\prod^{n-\tau}_{i=1}\sum_{j \in S_2}a_{ij} \cdot \sum_{ S_1 \subseteq \{1,...,n-\tau \}} (-1)^{|S_1|} \prod^{2n-\tau}_{i=n+1-\tau}\sum_{j \in S_1}a_{ij}\\
	     & = (-1)^{n} \sum_{ S_2 \subseteq \{n+1-\tau,...,2n-\tau \}  } (-1)^{|S_2|}\prod^{n-\tau}_{i=1}\sum_{j \in S_2}a_{ij} \cdot (-1)^{n-\tau}\sum_{ S_1 \subseteq \{1,...,n-\tau \}} (-1)^{|S_1|} \prod^{2n-\tau}_{i=n+1-\tau}\sum_{j \in S_1}a_{ij}\\
             & = \perm(\mA) \cdot \perm(\mA^{I}_{O})
 \end{align*}
\end{proof}

\begin{corollary}\label{c1} The following two observations can be made:
 \begin{enumerate} 
  \item An XOR-gadget itself could not be bipartite
  \item An EQUALITY-gadget itself could be bipartite
 \end{enumerate}
\end{corollary}

\begin{proof}
 Let $\mathfrak{G}$ be a bipartite graph $\mathfrak{G} =  \begin{pmatrix} 0 & \mA_\grG\\ \mA^\tp_\grG & 0 \end{pmatrix}$.
 Let $S_\mA = (r_1,r_2,r_3,r_4,r_5,r_6)$ be the signature of $\mA$. To make $\mathfrak{G}$ an XOR-gadget, its signature must be of the from (see Lemma \ref{ryser})
 \begin{equation}
	S_\mathfrak{G} = (0,0,\neq 0, \neq 0, 0,0) = (r_1^2, r_1r_2, r_1r_3, r_1r_4, r_1r_5, r_1r_6) = r_1S_\mA
 \end{equation}
 It follows that $r_1=0$ must be equal to $0$ modulo $p$ and hence also $r_1r_3$ and $r_1r_4$ must vanish modulo $p$, which contradicts the requirements of an XOR-gadget signature. \\

 The EQUALITY-gadget has the signature 
\begin{equation}
	S_\mathfrak{G} = (\neq 0, \neq 0, 0, 0, 0, 0) = (r_1^2, r_1r_2, r_1r_3, r_1r_4, r_1r_5, r_1r_6) = r_1S_\mA
 \end{equation}
 which leads to the fact that $r_1$ and $r_2$ must be non equal to zero modulo $p$ and hence $r_3,r_4,r_5,r_6$ must be also non zero, which does not contradict the requirements on an EQUALITY-gadget signature. 
\end{proof}

As the next step, we present our main conjecture, that states that certain gadgets can not be found.
\begin{conjecture}\label{tMain}
 Let $\grG$ be a graph with adjacency matrix $\mA$ and the signature
 \begin{equation}
  S_\grG = (r_1, r_2, r_3, r_4, r_5, r_6)_p
\end{equation}
for some prime $p > 2$. The bipartite double cover graph $\hat{\grG}$ can only be circular planar according to the nodes $(\text{dim}+i_1,\text{dim}+i_2) = (i'_1,i'_2), (o_1,o_2)$ 
if the mapping 
\begin{equation}
 \mu(r_i) = 	\begin{cases} 0 & ,\text{if } r_i = 0 \\
			      1 & ,\text{if } r_i > 0 \\
		\end{cases}
\end{equation}
does lead to a valid congruence of the form
 \begin{equation}\label{djiEquation}
  \mu(r_1) \cdot \mu(r_2) \equiv \mu(r_3) \cdot \mu(r_4) - \mu(r_5) \cdot \mu(r_6) \pmod{2}
 \end{equation}
An exception is the case where all $r_i > 0$. In that case, a circular bipartite double cover exists, but that is only \textbf{connectable} according to the nodes $i'_1,i'_2,o_1,o_2$.
\end{conjecture}

One could ``prove'' this conjecture with the sentence ``Otherwise $\classRP = \classNP$''. To do so, one would pick a XOR-gadget that contradicts the conjecture and take a formula $\Phi$ from $\udreisat$ and convert it into a pn-planar formula. Then one would construct the Valiant graph and build, according to Lemma \ref{lConnectable}, the bipartite double cover by simply substituting each gadgets by its bipartite double cover. Hence, the total bipartite double cover has the same planar layout as the Valiant graph and, abstracted, as the formulas incidence graph. After applying Kasteleyn's algorithm, one could determine $\Phi$ satisfiability as it is done in the proof of Theorem \ref{theorem2a}.

Note, the conjecture does not exclude the possibility that the bipartite double cover of a gadget will be planar, but only that it can not be  \textit{circular planar} according to all input/output nodes. We found gadgets that have three of the four nodes on the boundary but not all four. In Figure \ref{nestedXORgadgetBDC}, one can see a XOR-gadget that has one pair on the boundary and one pair in the interior, but an XOR-gadget with all four nodes on the boundary can not be found.

\begin{proposition}\label{c2} 
Since the signature $\mathcal{S} = (0,0,\neq 0, \neq 0,0,0)_p$ does not fulfill the Equation \eqref{djiEquation}, there can not exists an XOR-gadget $\gad$, such that $\hat{\gad}$ is circular planar with $(i'_1,i'_2)$ and $(o_1,o_2)$ as boundary nodes.
\end{proposition}

\begin{proposition}\label{c3}
 Corollary \ref{c1} together with Lemma \ref{lBIP} add up to the fact, that one can not find an \text{EQUALITY}-gadget $\gad$ that is itself bipartite (or any other bipartite gadget that could exists if $r_1$ is equal to one) and is circular planar, since its bipartite double cover would consists of two copies of $\gad$ and this would contradict Conjecture \ref{tMain}.
\end{proposition}


Let us turn to the one exception named in the Conjecture \ref{tMain}, that is the case when all $r_i > 0$. This exception is interesting, since it surprisingly allows the nodes  $(i'_1,i'_2)$ and $(o_1,o_2)$ to be on the boundary but puts restrictions to their order. On the first sight, one would argue, that a gadget with such a signature where no input variant vanishes modulo $p$ is not very useful, since it does not allow to exclude cyclic covers that would potential lead to non satisfying assignments. 

Before we proceed, we show the rules to concatenate two gadgets. We then show that if the exceptional gadgets from the conjecture would also be \textbf{extendable}, one could again build gadgets that contradict the conjecture itself. 


\begin{lemma}[Concatenating gadgets.] Given two gadgets $\gad_1$ and $\gad_2$ with the signatures $S_{\gad_1} = (r_1,r_2,r_3,r_4,r_5,r_6)$ and $S_{\gad_2} = (\rho_1,\rho_2,\rho_3,\rho_4,\rho_5,\rho_6)$
and the input and output nodes $(i_1,o_1),(i_2,o_2)$ and $(\iota_1,\omega_1),(\iota_2,\omega_2)$. If one concatenates them, by inserting the new edges $c_1 = (o_1,\iota_1)$ and $c_2 = (o_2,\iota_2)$, one gets a new gadget $\gad$ with the signature
\begin{align*}
 S_\gad = (r_1\rho_1, r_2\rho_2, r_3\rho_3+r_5\rho_6, r_4\rho_4+r_6\rho_5, r_3\rho_5+r_5\rho_4, r_4\rho_6+r_6\rho_3)
\end{align*}
and input/output pairs $(i_1,\omega_1),(i_2,\omega_2)$.
\end{lemma}

\begin{proof} The concatenation of two gadgets is shown in Figure \ref{connectingGadgets}.
 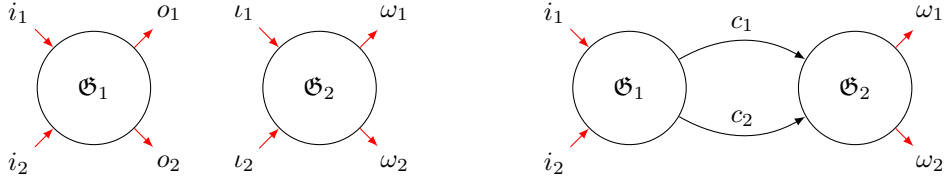
\begin{figure}[h!]
\centering
\begin{minipage}[h]{70mm}
	\centering
	\begin{tikzpicture}
		[scale=1]
		\tikzstyle{every loop} = [-latex]
		\node[circle, inner sep = 2pt] (OUT11) at (2,2){$o_1$};
		\node[circle, inner sep = 2pt] (OUT12) at (2,0){$o_2$};
		\node[circle, draw, inner sep = 10pt] (n1) at (1,1){$\gad_1$};
		\node[circle, inner sep = 2pt] (IN11) at (0,2){$i_1$};
		\node[circle, inner sep = 2pt] (IN12) at (0,0){$i_2$};
		
		\draw[-latex, color=red] (IN11) to (n1);
		\draw[-latex, color=red] (IN12) to (n1);
		\draw[-latex, color=red] (n1) to (OUT11);
		\draw[-latex, color=red] (n1) to (OUT12);

		\node[circle, inner sep = 2pt] (OUT21) at (5,2){$\omega_1$};
		\node[circle, inner sep = 2pt] (OUT22) at (5,0){$\omega_2$};
		\node[circle, draw, inner sep = 10pt] (n2) at (4,1){$\gad_2$};
		\node[circle, inner sep = 2pt] (IN21) at (3,2){$\iota_1$};
		\node[circle, inner sep = 2pt] (IN22) at (3,0){$\iota_2$};
		
		\draw[-latex, color=red] (IN21) to (n2);
		\draw[-latex, color=red] (IN22) to (n2);
		\draw[-latex, color=red] (n2) to (OUT21);
		\draw[-latex, color=red] (n2) to (OUT22);
		
	\end{tikzpicture}
\end{minipage}
\begin{minipage}[h]{70mm}
	\centering
	\begin{tikzpicture}
		[scale=1]
				\tikzstyle{every loop} = [-latex]
		\node[circle, draw, inner sep = 10pt] (n1) at (1,1){$\gad_1$};
		\node[circle, inner sep = 2pt] (IN11) at (0,2){$i_1$};
		\node[circle, inner sep = 2pt] (IN12) at (0,0){$i_2$};
		
		\draw[-latex, color=red] (IN11) to (n1);
		\draw[-latex, color=red] (IN12) to (n1);
		
		\node[circle, inner sep = 2pt] (OUT21) at (5,2){$\omega_1$};
		\node[circle, inner sep = 2pt] (OUT22) at (5,0){$\omega_2$};
		\node[circle, draw, inner sep = 10pt] (n2) at (4,1){$\gad_2$};
		
		\draw[-latex, color=red] (n2) to (OUT21);
		\draw[-latex, color=red] (n2) to (OUT22);
		
		\draw[-latex] (n1) to [out=30,in=150]  node[auto] {$c_1$} (n2);
		\draw[-latex] (n1) to [out=330,in=210]  node[auto] {$c_2$} (n2);

	\end{tikzpicture}
\end{minipage}
 \caption{This figure illustrates the concatenation principle of two extendable gadgets $\gad_1$ and $\gad_2$. The red edges are their external edges and their internal edges are not visible. The new gadget $\gad$ is constructed by inserting the edges $c_1$ and $c_2$, which are then internal edges of the new gadget $\gad$. The new input/output pairs are $(i_1,\omega_1),(i_2,\omega_2)$.}
\label{connectingGadgets}
\end{figure}
 A look at the figure immediatelly illustrates why only \textbf{extendable} gadgets can be used to create a new planar gadget. Assume that, e.g., at least one of those gadgets would be \textbf{connectable}. That means, in the case that it would be $\gad_1$, that the order of the boundary nodes would be $i_1,o_1,i_2,o_2$. The insertion of the new edges $(i_1,\omega_1),(i_2,\omega_2)$
would surround the input edge $i_2$, which would lead to an edge crossing and hence non-planarity.

 Concatenation of two gadgets is in principle the same as to create a small Valiant graph. Like Bon-Dor and Halevi argued in their proof \cite{Ben93}, we can argue with the independence of the two gadgets. Let the signature of $\gad$ be $S_\gad = (R_1,R_2,R_3,R_4,R_5,R_6)$, then 
 \begin{enumerate}
 	\item No external edge is taken and the permanent of $\gad$ is equal to the product of the permanent of $\gad_1$ and $\gad_2$. Hence $R_1 = r_1\rho_1$.
	\item All external edges are taken and thus also the two connecting edges $c_1$ and $c_2$ must be taken. It holds $\perm({(\mA_\gad)^{i_1,i_2}_{\omega_1,\omega_2}}) = \perm({(\mA_{\gad_1})^{i_1,i_2}_{o_1,o_2}})\perm({(\mA_{\gad_2})^{\iota_1,\iota_2}_{\omega_1,\omega_2}}) = r_2\rho_2 = R_2$.

	\item External edges $i_1$ and $\omega_1$ are taken. There are two possibilities to choose the connection edges: either $c_1$, therefore $o_1$ and $\iota_1$ are also taken, or $c_2$, thus also $o_2$ and $\iota_2$ are taken. Hence 
	\begin{align*}
		\perm({(\mA_\gad)^{i_1}_{\omega_1}}) 	& = \perm({(\mA_{\gad_1})^{i_1}_{o_1}})\perm({(\mA_{\gad_2})^{\iota_1}_{\omega_1}})+ \perm({(\mA_{\gad_1})^{i_1}_{o_2}})\perm({(\mA_{\gad_2})^{\iota_2}_{\omega_1}}) \\
							& = r_3\rho_3 + r_5\rho_6\\
							& = R_3
	\end{align*}
	\item External edges $i_2$ and $\omega_2$ are taken. There are two possibilities to choose the connection edges: either $c_2$, therefore $o_2$ and $\iota_2$ are also taken, or $c_1$, thus also $o_1$ and $\iota_1$ are taken. Hence 
	\begin{align*}
		\perm({(\mA_\gad)^{i_2}_{\omega_2}}) 	& = \perm({(\mA_{\gad_1})^{i_2}_{o_2}})\perm({(\mA_{\gad_2})^{\iota_2}_{\omega_2}})+ \perm({(\mA_{\gad_1})^{i_2}_{o_1}})\perm({(\mA_{\gad_2})^{\iota_1}_{\omega_2}}) \\
							& = r_4\rho_4 + r_6\rho_5\\
							& = R_4
	\end{align*}
	\item External edges $i_1$ and $\omega_2$ are taken. There are two possibilities to choose the connection edges: either $c_1$, therefore $o_1$ and $\iota_1$ are also taken, or $c_2$, thus also $o_2$ and $\iota_2$ are taken. Hence 
	\begin{align*}
		\perm({(\mA_\gad)^{i_1}_{\omega_2}}) 	& = \perm({(\mA_{\gad_1})^{i_1}_{o_1}})\perm({(\mA_{\gad_2})^{\iota_1}_{\omega_2}})+ \perm({(\mA_{\gad_1})^{i_1}_{o_2}})\perm({(\mA_{\gad_2})^{\iota_2}_{\omega_2}}) \\
							& = r_3\rho_5 + r_5\rho_4\\
							& = R_5
	\end{align*}
	\item External edges $i_2$ and $\omega_1$ are taken. There are two possibilities to choose the connection edges: either $c_2$, therefore $o_2$ and $\iota_2$ are also taken, or $c_1$, thus also $o_1$ and $\iota_1$ are taken. Hence 
	\begin{align*}
		\perm({(\mA_\gad)^{i_2}_{\omega_1}}) 	& = \perm({(\mA_{\gad_1})^{i_2}_{o_2}})\perm({(\mA_{\gad_2})^{\iota_2}_{\omega_1}})+ \perm({(\mA_{\gad_1})^{i_2}_{o_1}})\perm({(\mA_{\gad_2})^{\iota_1}_{\omega_1}}) \\
							& = r_4\rho_6 + r_6\rho_3\\
							& = R_6
	\end{align*}
 \end{enumerate}
\end{proof} 

The rules of concatenation explain, for example, what happens, if two gadgets, one with the signature $(1,2,2,2,2,2)_3$ and one with $(1,2,2,2,1,1)_3$, are concatenated. Note, that both gadgets are in principle possible due to the exception in Conjecture \ref{tMain}. The resulting gadget would have the signature $(1,1,0,0,0,0)_3$, which should not be possible, since it represents a EQUALITY-gadget. But in order to be able to extend one gadget with another to form a new gadget, that is still planar, the gadgets must have the property to be \textbf{extendable}. And Conjecture \ref{tMain} states, that all such exceptional gadgets, e.g. a gadget with a signature of $(1,2,2,2,2,2)_3$, are always only \textbf{connectable}, which makes the resulting new gadget non-planar.

\subsection{A solvable counting problem}
\textit{Question: Is there a counting problem, that can be solved with the gadgets that are realizable?}
 This question aims to an analogous case as to the $\#_7\mathsf{Pl}$-$\mathsf{RTW}$-$\mathsf{Mon}$-$\mathsf{3CNF}$ problem, that can ``accidentally'' be solved using holographic algorithms. 
 The answer is \textit{Yes, there is at least one}. One could count the solution of formulas from $\forestdreisat$, as defined in Definition \ref{defForSat}. Note, that this is not a new result, since in 2010 Samer and Szeider \cite{Sam10b} also showed how to count this efficiently, but using a totally different approach. The reason why this class is countable is, that one does not need circular planarity according to both input and output pairs but only for one of those pairs. The other pair must be located in a arbitrary inner face of the embedding. Since the incidence graph is a forest, all further gadgets could be stored \textit{within} the previous gadget, i.e., some kind of nesting has to be done. Since such gadgets do not contradict Conjecture \ref{tMain}, they can indeed be found. Consider the XOR-gadget in Figure \ref{nestedXORgadget} and its bipartite double cover in Figure \ref{nestedXORgadgetBDC}. Its bipartite double cover is connectable over the nodes $(8,3)$ and all further drawing is done within the nested area.

\begin{figure}[h!]
\centering
\begin{minipage}[h]{70mm}
	\centering
	\begin{tikzpicture}
		[scale=1]
		\tikzstyle{every loop} = [-latex]
		\node[circle, fill, draw, inner sep = 2pt] (n1) at (0,1.5){}; \node[] at (-0.3,1.5) {\tiny{$1$}};
		\node[circle, fill, draw, inner sep = 2pt] (n2) at (0.5,0){}; \node[] at (0.2,0) {\tiny{$2$}};
		\node[circle, fill, draw, inner sep = 2pt] (n3) at (0.5,3){}; \node[] at (0.2,3) {\tiny{$3$}};
		
		\node[circle, fill, draw, inner sep = 2pt] (n4) at (3,1.5){}; \node[] at (3.3,1.5) {\tiny{$4$}};
		\node[circle, fill, draw, inner sep = 2pt] (n5) at (2.5,0){}; \node[] at (2.8,0) {\tiny{$5$}};
		\node[circle, fill, draw, inner sep = 2pt] (n6) at (2.5,3){}; \node[] at (2.8,3) {\tiny{$6$}};

		\node[circle, fill, draw, inner sep = 2pt] (n7) at (1.5,1.5){}; \node[] at (1.3,1.7) {\tiny{$7$}};
		
		\draw[-latex] (n1) edge[bend left=10] (n2);
		\draw[-latex] (n1) edge[bend left=25] (n4);
		\draw[-latex] (n1) to (n7);
	
		\draw[-latex] (n2) edge[bend left=10] (n1);
		\draw[-latex] (n2) to (n3);
		\draw[-latex] (n2) edge[bend left=10] (n7);

		\draw[-latex] (n3) edge[bend left=10] (n6);

		\draw[-latex] (n4) to (n2);
		\draw[-latex] (n4) to (n3);
		\draw[-latex] (n4) edge[bend left=10] (n7);
	
		\draw[-latex] (n5) to (n1);
		\draw[-latex] (n5) edge[bend left=25] (n3);
		\draw[-latex] (n5) to [in=225,out=315, loop]  node[auto] {} (n5);
		\draw[-latex] (n5) edge[bend left=10] (n6);
		\draw[-latex] (n5) to (n7);

		\draw[-latex] (n6) edge[bend left=10] (n3);
		\draw[-latex] (n6) edge[bend left=10] (n5);
		
		\draw[-latex] (n7) edge[bend left=10] (n2);
		\draw[-latex] (n7) to (n3);
		\draw[-latex] (n7) edge[bend left=10] (n4);
		\draw[-latex] (n7) to [in=100,out=15, loop]  node[auto] {} (n7);
	\end{tikzpicture}
\end{minipage}
\begin{minipage}[h]{70mm}
	\centering
	$\begin{pmatrix}[r]
		0 & 1 & 0 & 1 & 0 & 0 & 1\\
		1 & 0 & 1 & 0 & 0 & 0 & 1\\
		0 & 0 & 0 & 0 & 0 & 1 & 0\\
		0 & 1 & 1 & 0 & 0 & 0 & 1\\
		1 & 0 & 1 & 0 & 1 & 1 & 1\\
		0 & 0 & 1 & 0 & 1 & 0 & 0\\
		0 & 1 & 1 & 1 & 0 & 0 & 1
	\end{pmatrix}$
\end{minipage}
\caption{$\mathfrak{G}^*: $The new XOR-gadget with (1,3) and (2,4).}
\label{nestedXORgadget}
\end{figure}
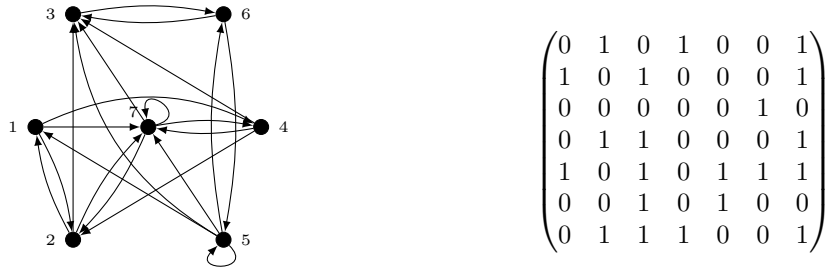

\begin{figure}[h!]
\centering
\begin{minipage}[h]{70mm}
	\centering
	\begin{tikzpicture}
		[scale=1]
		\tikzstyle{every loop} = [-latex]
		\node[circle, inner sep = 2pt, rotate=0] (help2) at (1,2.3) {\tiny{Nested Area}};
		\node[circle, inner sep = 2pt, rotate=90] (help) at (5.5,2.1) {...};
		\node[circle, inner sep = 2pt, rotate=90] (help2) at (2,2.3) {...};
	
		\node[circle, draw, inner sep = 57pt, dashed] (circ) at (2,2) {};
		\node[circle, fill, draw, inner sep = 2pt] (n1) at (2.8,0.8){}; \node[] at (3,1) {\tiny{$1$}};
		\node[circle, fill, draw, inner sep = 2pt] (n2) at (1,0){}; \node[] at (1,-0.3) {\tiny{$2$}};
		
		\node[circle, fill, draw, inner sep = 2pt] (n3) at (4,4){}; \node[] at (4,4.3) {\tiny{$3$}};
		\node[circle, fill, draw, inner sep = 2pt] (n4) at (2.5,3){}; \node[] at (2.3,3.2) {\tiny{$4$}};
		\node[circle, fill, draw, inner sep = 2pt] (n5) at (2.5,3.7){}; \node[] at (2.5,4) {\tiny{$5$}};
		\node[circle, fill, draw, inner sep = 2pt] (n6) at (1,4){}; \node[] at (1,4.3) {\tiny{$6$}};
		\node[circle, fill, draw, inner sep = 2pt] (n7) at (1.2,0.8){}; \node[] at (1.2,0.5) {\tiny{$7$}};
		\node[circle, fill, draw, inner sep = 2pt] (n8) at (4,0){}; \node[] at (4,-0.3) {\tiny{$8$}};
		\node[circle, fill, draw, inner sep = 2pt] (n9) at (2.5,1.6){}; \node[] at (2.8,1.6) {\tiny{$9$}};
		\node[circle, fill, draw, inner sep = 2pt] (n10) at (0,2.5){}; \node[] at (-0.2,2.3) {\tiny{$10$}};
		\node[circle, fill, draw, inner sep = 2pt] (n11) at (2,0.8){}; \node[] at (2,1.1) {\tiny{$11$}};
		\node[circle, fill, draw, inner sep = 2pt] (n12) at (2,4){}; \node[] at (2,4.3) {\tiny{$12$}};
		\node[circle, fill, draw, inner sep = 2pt] (n13) at (3,4){}; \node[] at (3,4.3) {\tiny{$13$}};
		\node[circle, fill, draw, inner sep = 2pt] (n14) at (3.5,0.4){}; \node[] at (3.8,0.4) {\tiny{$14$}};

		\draw[-latex, color=red] (n3) edge[bend left=40] (help);
		\draw[-latex, color=red] (help) edge[bend left=40] (n8);
		\draw[-latex, color=red] (n4) edge[bend right=40] (help2);
		\draw[-latex, color=red] (help2) edge[bend right=40] (n9);

		\draw[-] (n1) to (n9);\draw[-] (n1) to (n11);
		\draw[-] (n3) to (n13);
		\draw[-] (n13) to (n5);
		\draw[-] (n5) to (n12);\draw[-] (n5) edge[bend left=30] (n8);\draw[-] (n5) edge[bend right=20] (n10);
		\draw[-] (n12) to (n6);
		\draw[-] (n6) to (n10);
		\draw[-] (n9) to (n7);	
		\draw[-] (n7) to (n10); \draw[-] (n7) to (n14); 
		\draw[-] (n11) to (n7);		
		\draw[-] (n10) edge[bend right=30] (n2);
		\draw[-] (n2) to (n8);
		\draw[-] (n4) to (n9);\draw[-] (n4) edge[bend right=20] (n10);
		\draw[-] (n14) to (n5);\draw[-] (n14) to (n4);\draw[-] (n14) to (n2);\draw[-] (n14) to (n1);
	\end{tikzpicture}
\end{minipage}
\begin{minipage}[h]{80mm}
	\centering
	
	$\begin{pmatrix}[r r r r r r r r r r r r r r]
		0 & 0 & 0 & 0 & 0 & 0 & 0 & 0 & 1 & 0 & 1 & 0 & 0 & 1\\
		0 & 0 & 0 & 0 & 0 & 0 & 0 & 1 & 0 & 1 & 0 & 0 & 0 & 1\\
		0 & 0 & 0 & 0 & 0 & 0 & 0 & 0 & 0 & 0 & 0 & 0 & 1 & 0\\
		0 & 0 & 0 & 0 & 0 & 0 & 0 & 0 & 1 & 1 & 0 & 0 & 0 & 1\\
		0 & 0 & 0 & 0 & 0 & 0 & 0 & 1 & 0 & 1 & 0 & 1 & 1 & 1\\
		0 & 0 & 0 & 0 & 0 & 0 & 0 & 0 & 0 & 1 & 0 & 1 & 0 & 0\\
		0 & 0 & 0 & 0 & 0 & 0 & 0 & 0 & 1 & 1 & 1 & 0 & 0 & 1\\
		0 & 1 & 0 & 0 & 1 & 0 & 0 & 0 & 0 & 0 & 0 & 0 & 0 & 0\\
		1 & 0 & 0 & 1 & 0 & 0 & 1 & 0 & 0 & 0 & 0 & 0 & 0 & 0\\
		0 & 1 & 0 & 1 & 1 & 1 & 1 & 0 & 0 & 0 & 0 & 0 & 0 & 0\\
		1 & 0 & 0 & 0 & 0 & 0 & 1 & 0 & 0 & 0 & 0 & 0 & 0 & 0\\
		0 & 0 & 0 & 0 & 1 & 1 & 0 & 0 & 0 & 0 & 0 & 0 & 0 & 0\\
		0 & 0 & 1 & 0 & 1 & 0 & 0 & 0 & 0 & 0 & 0 & 0 & 0 & 0\\
		1 & 1 & 0 & 1 & 1 & 0 & 1 & 0 & 0 & 0 & 0 & 0 & 0 & 0
	\end{pmatrix}$
	
\end{minipage}
\caption{The bipartite double cover of $\gad^*$ with $(8,3)$ as boundary nodes and $(9,4)$ as internal input/output pair. Finding a gadgets with both pairs on the boundary to count arbitrary pn-planar formulas is not possible due to Conjecture \ref{tMain}. However, this gadget is sufficient to count solutions of formulas with incidence graphs that is a forest. }
\label{nestedXORgadgetBDC}
\end{figure}
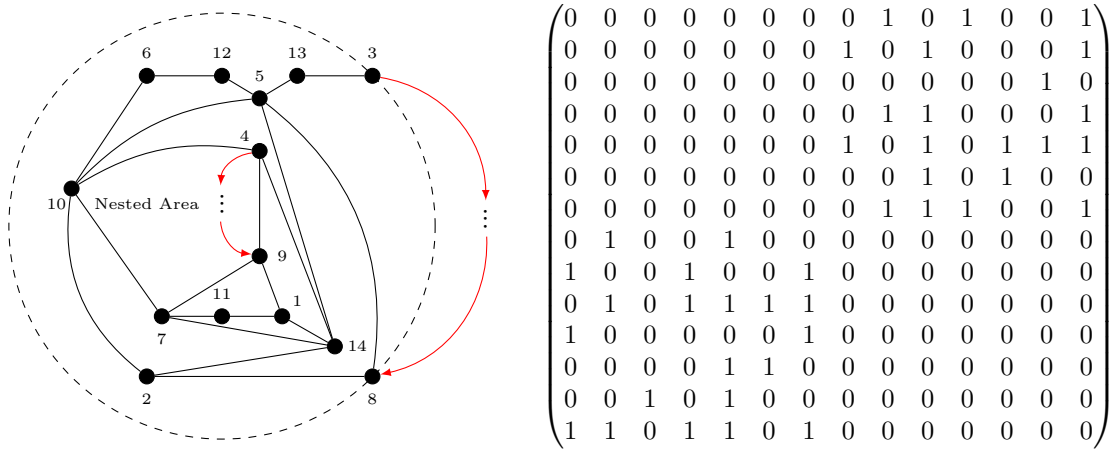
The signature of this XOR-gadget is
\begin{equation}
  S_{\gad^*} = (18, 3, 7, 8, 6, 3) = (0, 0, 1, 2, 0, 0)_3
\end{equation}

The CLAUSE-gadget does not have the same kind of signature like, e.g., the XOR-gadget. It has is a smaller kind of signature, that only care about the good cases, i.e., valid paths that can be taken by a cyclic cover. It must not care about what happens if the cyclic cover enters over input node $i_1$ but leaves about the non associated output node $o_3$. All this is handled by the XOR-gadget. For a CLAUSE-gadget with two outer edges, one could express this as the signature
\begin{equation}
 S = (\neq 0, 0, \neq 0, \neq 0, *, *)
\end{equation}
whereof $*$ stands for an arbitrary value. Similar arguments hold for the VARIABLE-gadget as well. 

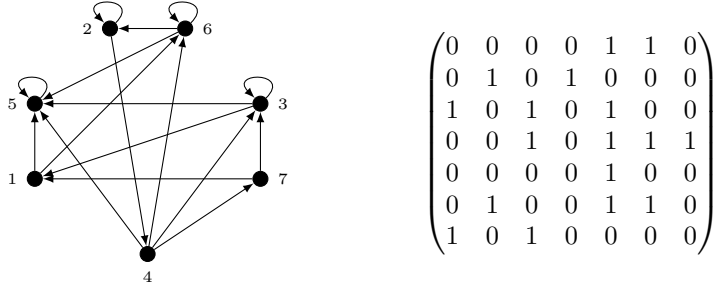
\begin{figure}[h!]
\centering
\begin{minipage}[h]{50mm}
	\begin{tikzpicture}
		[scale=1]
		\tikzstyle{every loop} = [-latex]
		\node[circle, fill, draw, inner sep = 2pt] (n1) at (0,1){}; \node[] at (-0.3,1) {\tiny{$1$}};
		\node[circle, fill, draw, inner sep = 2pt] (n5) at (0,2){}; \node[] at (-0.3,2) {\tiny{$5$}};
		\node[circle, fill, draw, inner sep = 2pt] (n6) at (2,3){}; \node[] at (2.3,3) {\tiny{$6$}};
		\node[circle, fill, draw, inner sep = 2pt] (n2) at (1,3){}; \node[] at (0.7,3) {\tiny{$2$}};
		\node[circle, fill, draw, inner sep = 2pt] (n4) at (1.5,0){}; \node[] at (1.5,-0.3) {\tiny{$4$}};
		\node[circle, fill, draw, inner sep = 2pt] (n7) at (3,1){}; \node[] at (3.3,1) {\tiny{$7$}};
		\node[circle, fill, draw, inner sep = 2pt] (n3) at (3,2){}; \node[] at (3.3,2) {\tiny{$3$}};
		\draw[-latex] (n1) to (n5);
		\draw[-latex] (n1) to (n6);
		
		\draw[-latex] (n2) to [in=135,out=45, loop]  node[auto] {} (n2);
		\draw[-latex] (n2) to (n4);
		
		\draw[-latex] (n3) to (n1);
		\draw[-latex] (n3) to [in=135,out=45, loop]  node[auto] {} (n3);
		\draw[-latex] (n3) to (n5);

		\draw[-latex] (n4) to (n3);
		\draw[-latex] (n4) to (n5);
		\draw[-latex] (n4) to (n6);
		\draw[-latex] (n4) to (n7);

		\draw[-latex] (n5) to [in=135,out=45, loop]  node[auto] {} (n5);

		\draw[-latex] (n6) to (n2);
		\draw[-latex] (n6) to (n5);
		\draw[-latex] (n6) to [in=135,out=45, loop]  node[auto] {} (n6);
		
		\draw[-latex] (n7) to (n1);
		\draw[-latex] (n7) to (n3);
		
	\end{tikzpicture}
\end{minipage}
\begin{minipage}[h]{50mm}
	\centering
	$\begin{pmatrix}[r]
	0 & 0 & 0 & 0 & 1 & 1 & 0\\
	0 & 1 & 0 & 1 & 0 & 0 & 0\\
	1 & 0 & 1 & 0 & 1 & 0 & 0\\
	0 & 0 & 1 & 0 & 1 & 1 & 1\\
	0 & 0 & 0 & 0 & 1 & 0 & 0\\
	0 & 1 & 0 & 0 & 1 & 1 & 0\\
	1 & 0 & 1 & 0 & 0 & 0 & 0
	\end{pmatrix}$
\end{minipage}
\caption{The new CLAUSE-gadget $\gad^{**}$ with its input and output pairs $(1,5), (2,6)$ and $(3,7)$.}
\label{newCLAUSEgadget}
\end{figure}

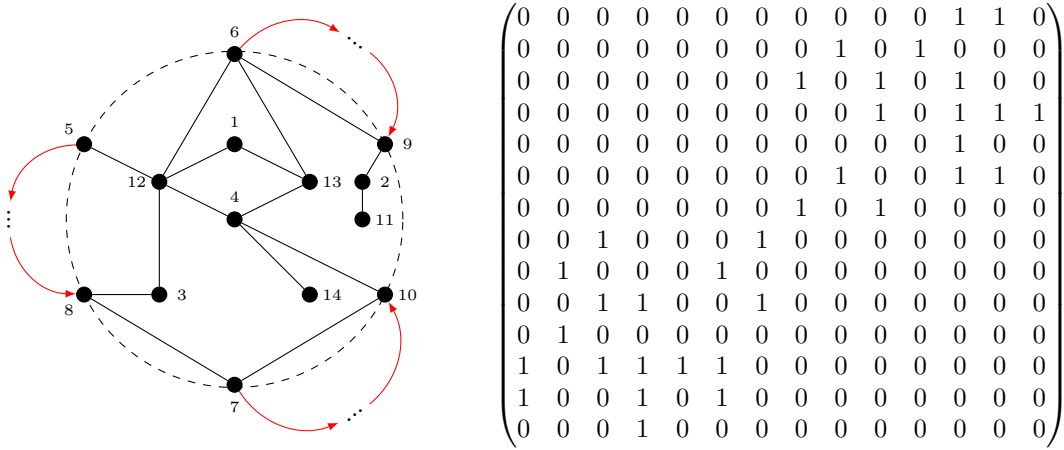
\begin{figure}[h!]
\centering
\begin{minipage}[h]{70mm}
	\centering
	\begin{tikzpicture}
		[scale=1]
		\tikzstyle{every loop} = [-latex]
		\node[circle, draw, inner sep = 45pt, dashed] (circ) at (2,2) {};
		\node[circle, fill, draw, inner sep = 2pt] (n4) at (0,3){}; \node[] at (-0.2,3.2) {\tiny{$5$}};
		\node[circle, fill, draw, inner sep = 2pt] (n5) at (2,4.2){}; \node[] at (2,4.5) {\tiny{$6$}};
		\node[circle, fill, draw, inner sep = 2pt] (n8) at (4,3){}; \node[] at (4.3,3) {\tiny{$9$}};
		
		\node[circle, fill, draw, inner sep = 2pt] (n7) at (0,1){}; \node[] at (-0.2,0.8) {\tiny{$8$}};
		\node[circle, fill, draw, inner sep = 2pt] (n6) at (2,-0.2){}; \node[] at (2,-0.5) {\tiny{$7$}};
		\node[circle, fill, draw, inner sep = 2pt] (n9) at (4,1){}; \node[] at (4.3,1) {\tiny{$10$}};

		\node[circle, fill, draw, inner sep = 2pt] (n0) at (2,3){}; \node[] at (2,3.3) {\tiny{$1$}};
		\node[circle, fill, draw, inner sep = 2pt] (n11) at (1,2.5){}; \node[] at (0.7,2.5) {\tiny{$12$}};
		\node[circle, fill, draw, inner sep = 2pt] (n12) at (3,2.5){}; \node[] at (3.3,2.5) {\tiny{$13$}};
		\node[circle, fill, draw, inner sep = 2pt] (n3) at (2,2){}; \node[] at (2,2.3) {\tiny{$4$}};

		\node[circle, fill, draw, inner sep = 2pt] (n2) at (1,1){}; \node[] at (1.3,1) {\tiny{$3$}};
		\node[circle, fill, draw, inner sep = 2pt] (n13) at (3,1){}; \node[] at (3.3,1) {\tiny{$14$}};

		\node[circle, fill, draw, inner sep = 2pt] (n1) at (3.7,2.5){}; \node[] at (4,2.5) {\tiny{$2$}};
		\node[circle, fill, draw, inner sep = 2pt] (n10) at (3.7,2){}; \node[] at (4,2) {\tiny{$11$}};

		\node[circle, inner sep = 2pt, rotate=90] (help1) at (-1,2) {...};
		\draw[-latex, color=red] (n4) edge[bend right=40] (help1);
		\draw[-latex, color=red] (help1) edge[bend right=40] (n7);

		\node[circle, inner sep = 2pt, rotate=-38] (help2) at (3.6,4.35) {...};
		\draw[-latex, color=red] (n5) edge[bend left=40] (help2);
		\draw[-latex, color=red] (help2) edge[bend left=40] (n8);

		\node[circle, inner sep = 2pt, rotate=38] (help2) at (3.6,-0.6) {...};
		\draw[-latex, color=red] (n6) edge[bend right=40] (help2);
		\draw[-latex, color=red] (help2) edge[bend right=40] (n9);

		\draw[-] (n1) to (n10);
		\draw[-] (n2) to (n7);
		\draw[-] (n3) to (n13);\draw[-] (n3) to (n9);
		\draw[-] (n4) to (n11);
		\draw[-] (n5) to (n11);\draw[-] (n5) to (n8);		
		\draw[-] (n5) to (n12);	
		\draw[-] (n6) to (n7);	
		\draw[-] (n6) to (n9);	
		\draw[-] (n8) to (n1);

		\draw[-] (n11) to (n0);\draw[-] (n11) to (n3);\draw[-] (n11) to (n2);		
		\draw[-] (n12) to (n0);\draw[-] (n12) to (n3);		
		
	\end{tikzpicture}
\end{minipage}
\begin{minipage}[h]{80mm}
	\centering
	
	$\begin{pmatrix}[r r r r r r r r r r r r r r]
		0 & 0 & 0 & 0 & 0 & 0 & 0 & 0 & 0 & 0 & 0 & 1 & 1 & 0\\
		0 & 0 & 0 & 0 & 0 & 0 & 0 & 0 & 1 & 0 & 1 & 0 & 0 & 0\\
		0 & 0 & 0 & 0 & 0 & 0 & 0 & 1 & 0 & 1 & 0 & 1 & 0 & 0\\
		0 & 0 & 0 & 0 & 0 & 0 & 0 & 0 & 0 & 1 & 0 & 1 & 1 & 1\\
		0 & 0 & 0 & 0 & 0 & 0 & 0 & 0 & 0 & 0 & 0 & 1 & 0 & 0\\
		0 & 0 & 0 & 0 & 0 & 0 & 0 & 0 & 1 & 0 & 0 & 1 & 1 & 0\\
		0 & 0 & 0 & 0 & 0 & 0 & 0 & 1 & 0 & 1 & 0 & 0 & 0 & 0\\
		0 & 0 &	1 & 0 &	0 & 0 & 1 & 0 & 0 & 0 & 0 & 0 & 0 & 0\\
		0 & 1 &	0 & 0 &	0 & 1 & 0 & 0 & 0 & 0 & 0 & 0 & 0 & 0\\
		0 & 0 &	1 & 1 & 0 & 0 & 1 & 0 & 0 & 0 & 0 & 0 & 0 & 0\\
		0 & 1 &	0 & 0 &	0 & 0 & 0 & 0 & 0 & 0 & 0 & 0 & 0 & 0\\
		1 & 0 &	1 & 1 &	1 & 1 & 0 & 0 & 0 & 0 & 0 & 0 & 0 & 0\\
		1 & 0 &	0 & 1 &	0 & 1 & 0 & 0 & 0 & 0 & 0 & 0 & 0 & 0\\
		0 & 0 &	0 & 1 &	0 & 0 & 0 & 0 & 0 & 0 & 0 & 0 & 0 & 0
	\end{pmatrix}$
	
\end{minipage}
\caption{The bipartite double cover of the new CLAUSE-gadget with the input output pairs $(8,5),(9,6)$ and $(10,7)$.}
\label{newCLAUSEgadgetBDC}
\end{figure}

\begin{theorem}\label{theorem2c}
 Let $\Phi$ be random boolean formula from $\#\forestdreisat$ with $n$ variables and $m$ clauses and let $\grG$ be the graph that is constructed from $\Phi$ according to Valiant's proof, using the gadgets $\mathfrak{G}^{**}$ (see Figure \ref{clauseGadget}), $\mathfrak{G}_2$ (see Figure \ref{variableGadget}), $\mathfrak{G}^*$ (see Figure \ref{newXORgadget}). Then there exists an algorithm $\mathcal{A}$ that decides $\Phi$'s satisfiability in randomized polynomial time $\mathcal{O}(km^3)$ with an a success probability of 
\begin{equation*}
	\mathsf{Pr}[\mathsf{s} \leftarrow \mathcal{A}(\Phi,k) | \mathsf{1}_{\Phi\;\mathrm{is\;sat}} = \mathsf{s} ] = 1 - \mathsf{s}\frac{1}{3^k} 
\end{equation*}
\end{theorem}

\begin{proof}
 Note that a formula $\Phi$ from $\forestdreisat$ is always planar. Figure \ref{forestGraph} shows the idea of how to construct such a graph.  W.l.o.g. we assume that the incidence graph of $\Phi$ 
is connected and we start at an arbitrary clause node. Further note that the orientation of the external edges of each gadget (i.e. clockwise or anticlockwise) does not matter, since we can turn each connecting gadget into the correct orientation by flipping it around. This is due to the fact that the formula is from $\forestdreisat$. Otherwise, we would have to find several gadgets for each type (like the gadget $\gad_{\text{RL}}$ and $\gad_{\text{LR}}$) to take care of the orientations.

We now begin to redraw the incidence graph of the formula $\Phi$ using our new graph gadgets. For the current clause node we insert the new CLAUSE-gadget $\gad^{**}$. This CLAUSE-gadget has thee direction and in each direction there is a tree. W.l.o.g. we follow one arbitrary direction. In the example graph, e.g., assume we take the direction from $c_1$ to the variable node $1$. Next, we pick the XOR-gadget $\mathfrak{G}^*$, turn it into the right direction and connect it to the CLAUSE-gadget. For the interior input/output pair of the XOR-gadget, we pick the VARIABLE-gadget $\mathfrak{G}_2$, turn into the right orientation and connect it to represent the variable $1$ in the graph. Following the whole formula $\Phi$ that way and connecting the gadgets as described, allows to obtain a bipartite double cover of $\grG$ and hence to apply Kasteleyn's algorithm. Regarding this gadgets the total graph will be of

\begin{align*}
	m|\mathfrak{G}^{**}| + \sum^n_{i=1}|\mathfrak{G}_2(x_i)| +  3m|\mathfrak{G}^*| 	& = m\cdot 7 + \sum^n_{i=1}(2+2n_i) + 3m\cdot 7 \\
											& = 2n+2\sum^n_{i=1}n_i + 28m \\
											& < 6m + 6m + 28m \\
											& < 40m 
\end{align*}
nodes. The bipartite double cover operations doubles the number of nodes, so we can bound the number of nodes from above via $|\hat{\grG}| \leq 80m$. Finally, Kasteleyn's algorithms runs in time 
that is cubic in the number of nodes, which yields a total running time of our algorithm of $\mathcal{O}(80^3m^3)$, or $\mathcal{O}(m^3)$ for large $m$, for a single evaluation. Applying algorithm $\mathcal{B}$ from the proof of Theorem \ref{theorem2b} yields a total runtime of $\mathcal{O}(km^3)$.

\begin{figure}[h!]
\begin{minipage}[h]{50mm}
	\begin{tikzpicture}
		\centering
		[scale=1]
		\tikzstyle{every loop} = [-latex]
		\node[circle, draw, inner sep = 4pt, thick] (c1) at (0,2){\tiny{$c_1$}}; 
		\node[circle, draw, inner sep = 2pt] (v1) at (0,1){\tiny{$1$}}; 
		\node[circle, draw, inner sep = 4pt, thick] (c2) at (1,1){\tiny{$c_2$}};
		\node[circle, draw, inner sep = 2pt] (v2) at (2,1){\tiny{$2$}};
		\node[circle, draw, inner sep = 2pt] (v3) at (1,0){\tiny{$3$}};
		\node[circle, draw, inner sep = 4pt, thick] (c3) at (1,-1){\tiny{$c_3$}};	
		\draw[-] (c1) to (v1);
		\draw[-] (v1) to (c2);
		\draw[-] (c2) to (v2);
		\draw[-] (c2) to (v3);
		\draw[-] (v3) to (c3);
	\end{tikzpicture}
\end{minipage}
\begin{minipage}[h]{50mm}
	\begin{tikzpicture}
		\centering
		[scale=1]
		\tikzstyle{every loop} = [-latex]
		\node[circle, draw, inner sep = 4pt, thick] (c1) at (0,2){\tiny{$c_1$}}; 
			\node[rectangle, draw, minimum width = 90pt, minimum height = 95pt, rounded corners=2pt] (R1) at (1.2,-0.2) {};
		\node[circle, draw, inner sep = 2pt] (v1) at (0,1){\tiny{$1$}}; 
			\node[rectangle, draw, minimum width = 65pt, minimum height = 90pt, rounded corners=2pt] (R2) at (1.6,-0.2) {};
		\node[circle, draw, inner sep = 4pt, thick] (c2) at (1,1){\tiny{$c_2$}};
			\node[rectangle, draw, minimum width = 22pt, minimum height = 18pt, rounded corners=2pt] (R3) at (2,1) {};
		\node[circle, draw, inner sep = 2pt] (v2) at (2,1){\tiny{$2$}};
			\node[rectangle, draw, minimum width = 27pt, minimum height = 60pt, rounded corners=2pt] (R4) at (1,-0.65) {};
		\node[circle, draw, inner sep = 2pt] (v3) at (1,0){\tiny{$3$}};
			\node[rectangle, draw, minimum width = 23pt, minimum height = 30pt, rounded corners=2pt] (R5) at (1,-1) {};
		\node[circle, draw, inner sep = 4pt, thick] (c3) at (1,-1){\tiny{$c_3$}};	
		\draw[-] (c1) to (v1);
		\draw[-] (v1) to (c2);
		\draw[-] (c2) to (v2);
		\draw[-] (c2) to (v3);
		\draw[-] (v3) to (c3);
	\end{tikzpicture}
\end{minipage}
\caption{On the left is the incidence graph of a small formula from $\forestdreisat$. On the right its drawing using our proof idea. The rectangles illustrates the XOR-gadget that is between each clause and variable gadget regarding the edge it crosses. Because of the forest property, all further gadgets can be nested into the latest XOR-gadget.}
\label{forestGraph}
\end{figure}
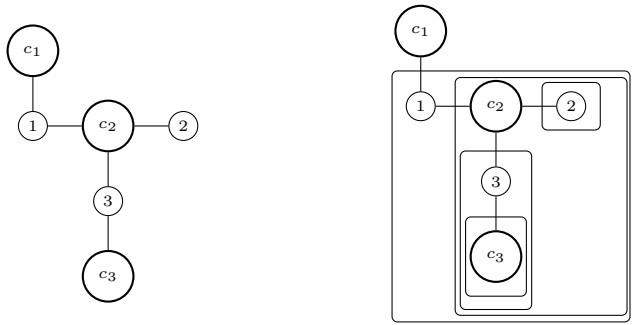

\end{proof}

\section{Concluding remarks} The paper shows up a general approach to count the solutions of a boolean planar formula if one finds appropriate gadgets, i.e. that do not contradict Conjecture \ref{tMain}, to build the Valiant graph. Currently, we only showed that $\#\forestdreisat$ could be solved with this approach, but we only restrict our attention to a graph that is build in the same way as Valiant build its graph in his proof. Are there more special formula classes that could be addressed using another type of Valiant graph? 

The conjecture itself is based on many computations and practical verifications and not at least on the very well believed assumption that $\classRP \neq \classNP$. Proving this conjecture is another open task for the future work in this direction.

\section{Acknowledgements}

Thanks goes to all my colleagues for fruitful discussion about the topic and to Peter Shor, who pointed out to me the existence of the Desnanot-Jacobi identity by answering one of my questions on
\textit{cstheory.stackexchange.com}.

\bibliography{citations}
\bibliographystyle{plain}

\end{document}